\documentclass[sigconf]{aamas} 

\usepackage{balance} 
\usepackage{xspace}
\usepackage{color}
\usepackage{algorithmic}
\usepackage[ruled,vlined,linesnumbered]{algorithm2e}
\usepackage{epsfig}
\usepackage{subfig}
\usepackage{mathrsfs,amsmath,amsthm,mdwlist,bbm}

\usepackage{makecell}
\usepackage{multirow}

\allowdisplaybreaks

\renewcommand{\leq}{\leqslant}
\renewcommand{\geq}{\geqslant}
\renewcommand{\ge}{\geqslant}
\renewcommand{\le}{\leqslant}

\usepackage{enumerate}

\usepackage[shortlabels]{enumitem}

 \newtheorem{example}{\textbf{Example}}
\newtheorem{thm}{Theorem}
\newtheorem{corollary}[thm]{Corollary}
\newtheorem{lemma}{Lemma}
\newtheorem{definition}{Definition}

 \setlist{nolistsep,leftmargin=*}

\usepackage{nicefrac}

\newcommand{\eps}{\ensuremath{\varepsilon}\xspace}
\renewcommand{\epsilon}{\eps}
\let\mydelta\delta
\renewcommand{\delta}{\ensuremath{\mydelta}\xspace}
\let\myalpha\alpha
\renewcommand{\alpha}{\ensuremath{\myalpha}\xspace}
\let\mystar\star
\renewcommand{\star}{\ensuremath{\mystar}}

\newcommand{\SC}{{\sc Set Cover}\xspace}
\newcommand{\TMCV}{{\sc TMCV}\xspace}

\newcommand{\pr}{\ensuremath{\prime}}

\newcommand{\el}{\ensuremath{\ell}\xspace}

\newcommand{\NPH}{\ensuremath{\mathsf{NP}}-hard\xspace}
\newcommand{\NPC}{\ensuremath{\mathsf{NP}}-complete\xspace}

\newcommand{\YES}{{\sc{yes}}\xspace}

\newcommand{\EE}{\ensuremath{\mathcal E}\xspace}

\newcommand{\GG}{\ensuremath{\mathcal G}\xspace}

\newcommand{\OO}{\ensuremath{\mathcal O}\xspace}

\renewcommand{\SS}{\ensuremath{\mathcal S}\xspace}
\newcommand{\TT}{\ensuremath{\mathcal T}\xspace}
\newcommand{\UU}{\ensuremath{\mathcal U}\xspace}

\newcommand{\WW}{\ensuremath{\mathcal W}\xspace}

\newtheorem{observation}{\bf Observation}

\usepackage{cleveref}

\crefname{theorem}{Theorem}{Theorems}
\crefname{observation}{Observation}{Observations}
\crefname{lemma}{Lemma}{Lemmas}
\crefname{corollary}{Corollary}{Corollaries}
\crefname{proposition}{Proposition}{Propositions}
\crefname{definition}{Definition}{Definitions}
\crefname{claim}{Claim}{Claims}
\crefname{table}{Table}{Tables}
\crefname{equation}{Inequality}{Inequalities}
\crefname{reductionrule}{Reduction rule}{Reduction rules}
\crefname{section}{Section}{Sections}

\setcopyright{ifaamas}
\acmConference[AAMAS '21]{Proc.\@ of the 20th International Conference on Autonomous Agents and Multiagent Systems (AAMAS 2021)}{May 3--7, 2021}{London, UK}{U.~Endriss, A.~Now\'{e}, F.~Dignum, A.~Lomuscio (eds.)}
\copyrightyear{2021}
\acmYear{2021}
\acmDOI{}
\acmPrice{}
\acmISBN{}

\acmSubmissionID{122}

\title{Network Robustness via Global $k$-cores}

\author{Palash Dey$^*$}\thanks{$^*$The list of authors is in alphabetical order.}
\affiliation{%
  \institution{Indian Institute of Technology, Kharagpur}
}
\email{palash.dey@cse.iitkgp.ac.in}
\author{Suman Kalyan Maity$^*$}
\affiliation{%
  \institution{Northwestern University}
}
\email{suman.maity@kellogg.northwestern.edu}
\author{Sourav Medya$^*$}
\affiliation{%
  \institution{Northwestern University}
}
\email{sourav.medya@kellogg.northwestern.edu}

\author{Arlei Silva$^*$}
\affiliation{%
  \institution{University of California Santa Barbara}
}
\email{arlei@cs.ucsb.edu}

\begin{abstract}
Network robustness is a measure a network's ability to survive adversarial attacks. But not all parts of a network are equal. K-cores, which are dense subgraphs, are known to capture some of the key properties of many real-life networks. Therefore, previous work has attempted to model network robustness via the stability of its $k$-core. However, these approaches account for a single core value and thus fail to encode a global network resilience measure. 
 In this paper, we address this limitation by proposing a novel notion of network resilience that is defined over all cores. In particular, we evaluate the stability of the network under node removals with respect to each node's initial core. Our goal is to compute robustness via a combinatorial problem: find $b$ most critical nodes to delete such that the number of nodes that fall from their initial cores is maximized. One of our contributions is showing that it is NP-hard to achieve any polynomial factor approximation of the given objective. We also present a fine-grained complexity analysis of this problem under the lens of parameterized complexity theory for several natural parameters. Moreover, we show two applications of our notion of robustness: measuring the evolution of species and characterizing networks arising from different domains.
 
\end{abstract}

\keywords{Network optimization, k-core, network robustness}

\newcommand{\BibTeX}{\rm B\kern-.05em{\sc i\kern-.025em b}\kern-.08em\TeX}

\begin{document}


\pagestyle{fancy}
\fancyhead{}


\maketitle 


\section{Introduction}

Networks model many real-world complex systems. An important aspect of these networks is their robustness or resilience. Robustness quantifies a network's capability to resist failures that might affect its functionalities. These network failures often lead to a considerable of economic losses. As an example, a snowy weather in 2008 caused a major power grid failure in China \cite{zhou2016two}.


The study of network resilience via stability of the $k$-core structure \cite{seidman1983network} has been a popular topic in recent literature. Bhawalkar et al. \cite{bhawalkar2015preventing} propose maximizing the initial $k$-core size to prevent network unravelling.
The resilience of $k$-core have also been studied under critical node/edge deletion to increase or maintain users' engagement in social networks \cite{zhang2017finding,zhou2019k,medya2019k} and to prevent failures in technological networks \cite{Laishram2018}. Consider an example of a P2P network where the users who benefit from the network should also share their resources with other users. This follows a $k$-core model and in this case the network owner has to be aware of the critical nodes to maintain the resource sharing process uninterrupted. 



The aforementioned studies suffer from a local notion of network stability as they aim to modify the $k$-core for a given value of $k$. We address this limitation by proposing a novel combinatorial problem over $k$-cores: \textit{find $b$ (budget) critical nodes whose deletion will remove the maximum number of nodes from their initial core}. The number of nodes staying in their core after removal of those critical nodes quantifies the stability of the network. Thus, a network is more (less) robust or resilient if a larger set of nodes are unaffected (affected). 



\begin{figure}[t]
    \centering
    \small
    \subfloat[Random ]{\includegraphics[width=0.22\textwidth]{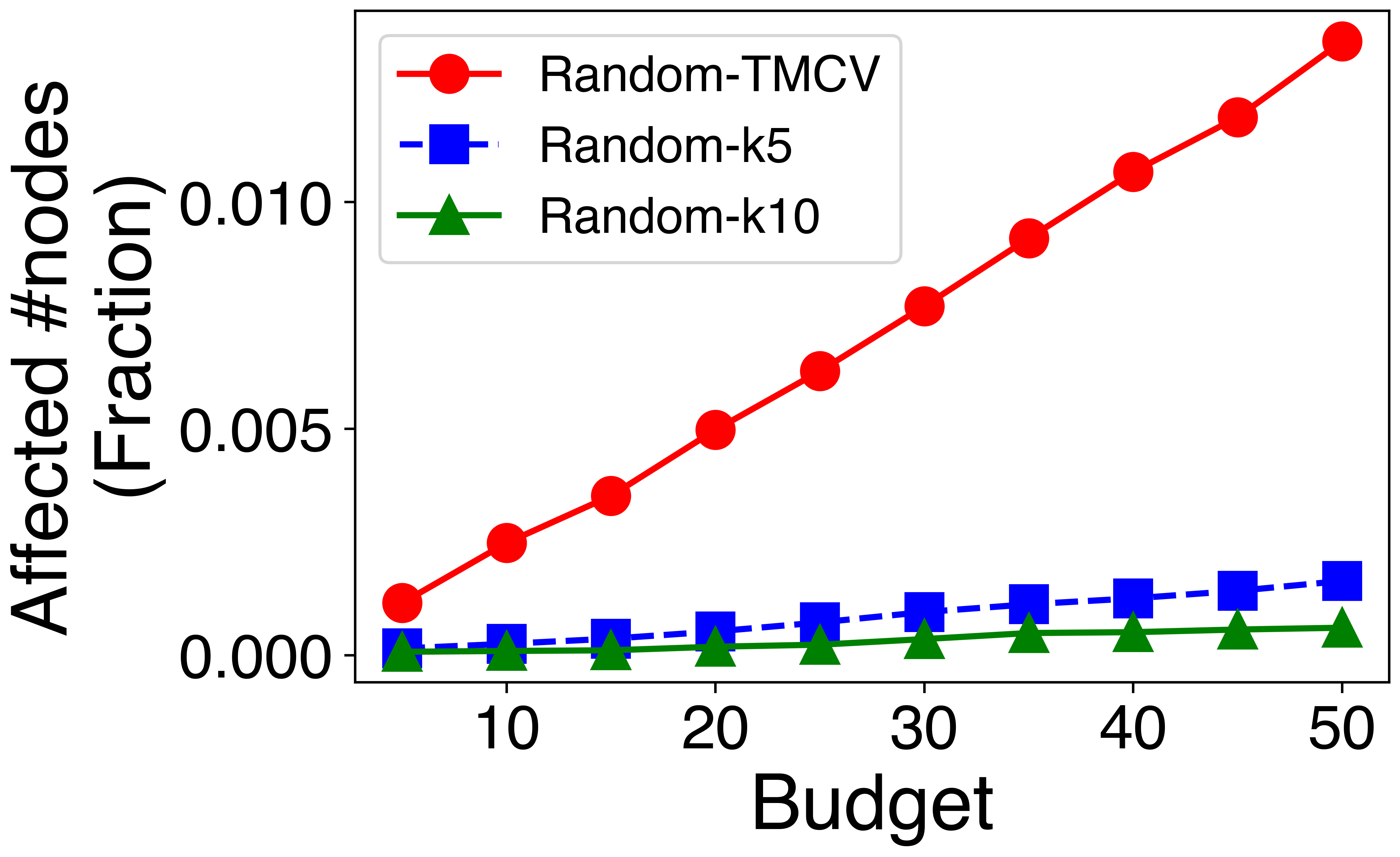}\label{fig:random}}
    \hspace{4mm}
     \subfloat[High Degree ]{\includegraphics[width=0.22\textwidth]{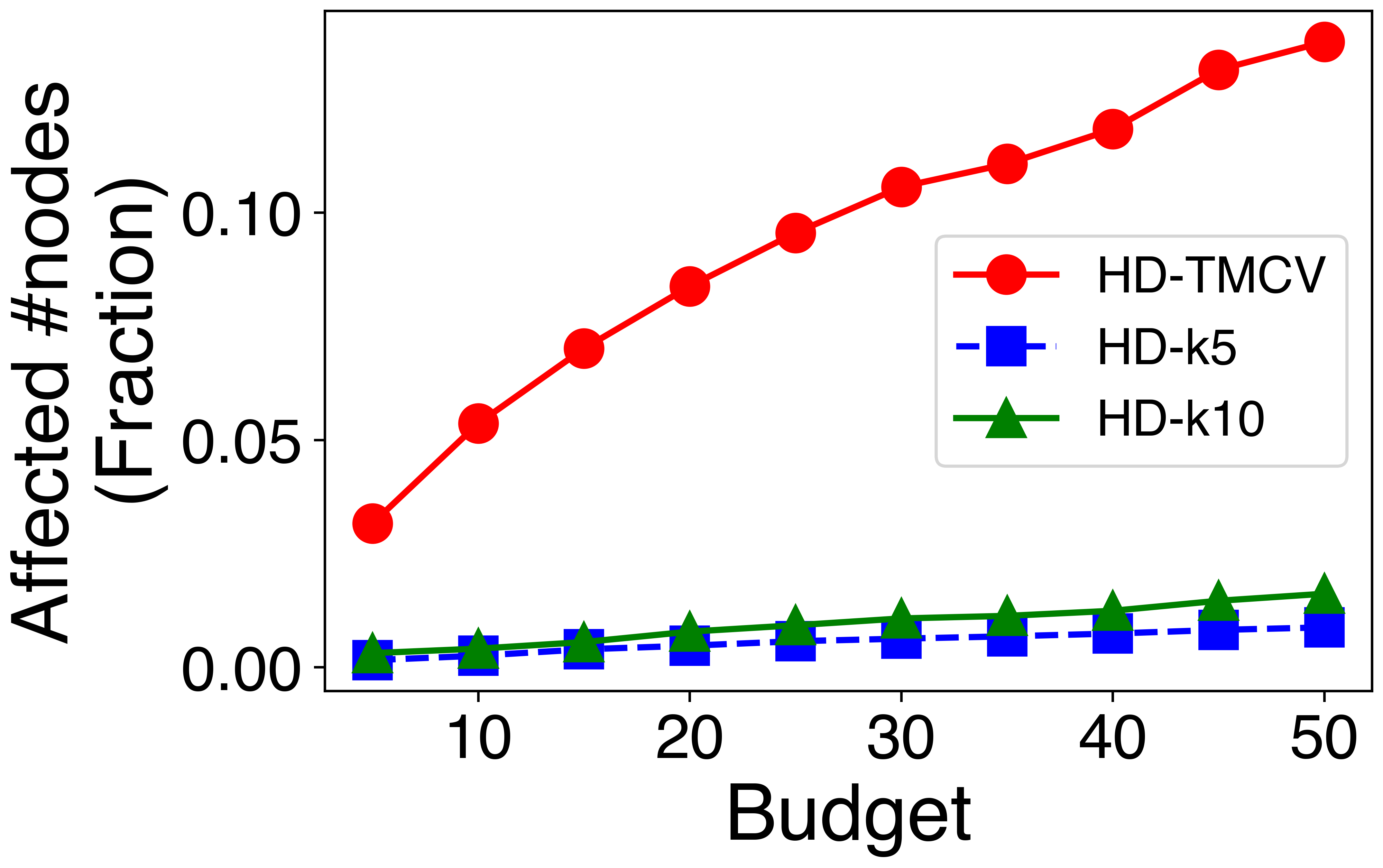}\label{fig:degree}}
    \caption{ \small Fraction ($F_v$) of the entire node set is affected by deleting (a) random and (b) high degree nodes. The red line shows the result by our objective (see TMCV in Def. \ref{def:tmcv}), whereas, the other ones show the effect inside fixed $k$-cores ($k=5$ and $k=10$) in a co-authorship network (CondMat in Sec. \ref{sec:exp}). Our objective has a global and larger impact on the entire network. \label{fig:problems}}
\end{figure}

Figure \ref{fig:problems} shows an example of the global effect of our formulation. We show how the nodes get affected (i.e., fall from their initial core) under node deletion via two strategies when (a) random and (b) high degree nodes are selected. The $y$-axis shows the fraction of the total nodes that get affected. Our formulated objective (red line) captures a global robustness notion and the number of affected nodes are much larger than in the individual cores (denoted by blue and green for $5$-core and $10$-core respectively). 

\subsection{Contributions}

We study a novel combinatorial problem, \textit{Total Minimization of Coreness via Vertex deletion} (\TMCV), which aims to measure network robustness based on the maximum number of nodes that fall from their initial core after $b$ number of nodes are deleted. 
Besides showing strong inapproximability result, we present fine-grained parameterized complexities of the problem for several parameters. Table \ref{tab:tmcv_summary} (in Section \ref{sec:theory}) summarizes the main theoretical results.

Additionally, we propose a few heuristics to solve \TMCV and evaluate their performance on real datasets. These heuristics nicely capture interesting structural properties of networks from different genres (e.g., social, co-authorship). Furthermore, we apply our proposed network robustness measure to understand the evolution of species. Zitnik et al. \shortcite{Zitnik454033} has shown that evolution is related to a network robustness measure based on network connectivity. Intuitively, more genetic changes in a species would result in a more resilient protein-protein interaction network of the same. In Section \ref{sec:exp}, we show significant correlation between our proposed resilience/robustness measure and the evolution dynamics of species. 


Our main contributions are as follows:

\begin{itemize}[itemsep=.2cm]
 \item We propose a novel network robustness problem (TMCV) based on the coreness of nodes under deletion of nodes.
 
 \item We show that it is NP-hard to achieve any polynomial factor approximation for \TMCV (Thm. \ref{thm:inapprox_tmcv}).

 \item We study the parameterized complexity of our problem for several natural parameters. We show that \TMCV is $W[2]$-hard (Thm. \ref{thm:TMCV_param_b}) parameterized by the budget $b$ and para-NP-hard parameterized by the degeneracy (Cor. \ref{cor:tmcv_para_d}) of the graph and the maximum degree (Thm. \ref{thm:tmcv_para_deg}) individually. 
 
 \item We propose several heuristics that capture interesting structural properties of networks from different genres. Furthermore, we show how we can apply our network robustness measure to understand the evolution of species. 
\end{itemize}

\paragraph{Organization of the paper} The paper is organized as follows: Section \ref{sec:related_work} describes the related work. We define our network robustness problem in Section \ref{sec:prob_def}. We show how to apply our network robustness measure to capture interesting structural properties of networks as well as to understand the evolution of species in Section \ref{sec:exp}. Finally, Section \ref{sec:theory} demonstrates all the theoretical results. 

\subsection{Related Work}
\label{sec:related_work}

Understanding robustness of a network via the stability of its $k$-core has recently received a significant amount of attention. The major goal in this line of work is to measure the resilience of the $k$-core of a network under its modifications. Zhang et al. \cite{zhang2017finding} first propose the collapsed $k$-core problem that aims to minimize the $k$-core by deleting $b$ critical vertices. The edge version of this problem has been recently addressed with efficient heuristics \cite{zhu2018k,medya2019k}. Another related paper \cite{Laishram2018} measures the stability of $k$-core under random edge/node deletions. These studies only focus on the $k$-core robustness, i.e., the effect on the nodes inside the $k$-core. 
On the contrary, we propose a novel and generalized version of these problems. Our robustness measure captures the affected nodes in different cores (i.e., any $k$) upon a budget number of node deletions.


Other related but orthogonal literature studies the  maximization of the $k$-core in networks via different mechanisms. One such example is maximization of the $k$-core by making a few nodes outside the $k$-core as anchors to prevent unraveling in social networks \cite{bhawalkar2015preventing,chitnis2013preventing}. The other example involves adding edges with nodes outside of the $k$-core \cite{zhou2019k}. Another related paper \cite{luo2018parameterized} discusses parameterized algorithms for the collapsed k-core problem \cite{zhang2017finding}. In this paper, we discuss parameterized complexity for a different problem along with inapproximabilty results.

\textbf{Network robustness: } 
Previous work has also studied the ability of a network to sustain various types of attacks or failures and termed it as network robustness. An extensive survey of different robustness measures for undirected and unweighted networks is conducted by Ellens et al. \cite{ellens2013graph}. The measures discussed in this survey are mainly based on network connectivity and shortest path distances along with some others based on Laplacian eigenvalues. Another survey \cite{liu2017comparative} provides nine widely used robustness measures and studies their sensitivity. The network robustness models vary depending on the application. In a recent work, Lordan et al. \cite{lordan2019exact} have identified the set of removed nodes that maximizes the size of the largest connected component of a network in an optimal manner. Recently, Zitnik et al. \cite{Zitnik454033} show that a connectivity-based robustness of protein-protein interaction networks is a good predictor of the extent of evolution of a species. Here, we consider a different robustness measure.  

Our proposed problem is also related to network design problems. These problems aim to optimize network properties or processes under network modifications. Examples include diameter \cite{demaine2010}, node centrality \cite{crescenzi2015,medya2018group}  shortest path~\cite{meyerson2009,dilkina2011,medya2018noticeable,medya2018making}, and influence spread \cite{Tong2012GML,Khalil2014,kimura2008minimizing,medya2020approximate} improvement. However, our objective is different from the ones considered by network design studies. 


\section{Problem Definition}
\label{sec:prob_def}

Let $G(V,E)$  be an undirected and unweighted graph with sets of vertices $V$ ($|V|=n$) and edges $E$ ($|E|=m$). We denote the degree of vertex $u$ in $G$ by $d(u,G)$. An induced subgraph, $H=(V_H,E_H)$ of $G$ is the following: if $u,v \in V_H$ and $(u,v)\in E$ then $(u,v)\in E_H$. The $k$-core \cite{seidman1983network} of a network is defined as follows.

\begin{definition} \textbf{$k$-Core:} The $k$-core of a graph $G$, denoted by $S_k(G)=(V_k(G),E_k(G))$, is defined as a maximal induced subgraph where each vertex has degree at least $k$.
\end{definition}

\begin{definition} \textbf{Coreness:} The coreness of a node $v$ in graph $G$, denoted by $C(v,G)$, is defined as the maximum $k$ where $v\in S_k(G)$ and $v\notin S_l(G)$ for any $l>k$.
\end{definition}

\begin{definition} \textbf{Degeneracy:} The degeneracy  of a graph $G$, denoted by $D(G)$, is defined as the largest $k$ where $S_k(G)$ is non-empty.
\end{definition}

\begin{example}
Consider the initial graph in Figure \ref{fig:init_graph} as an example. The degeneracy of the graph is $3$.
 \end{example}
 
 \begin{figure}[t]
 \vspace{-3mm}
    \centering
    \small
    \subfloat[Initial ]{\includegraphics[width=0.12\textwidth]{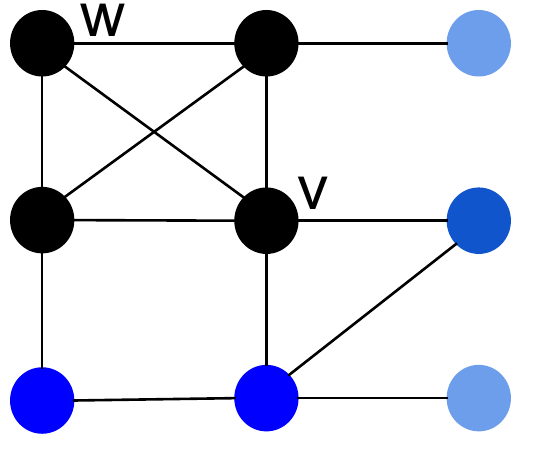}\label{fig:init_graph}}
    \subfloat[Only 3-core]{\includegraphics[width=0.12\textwidth]{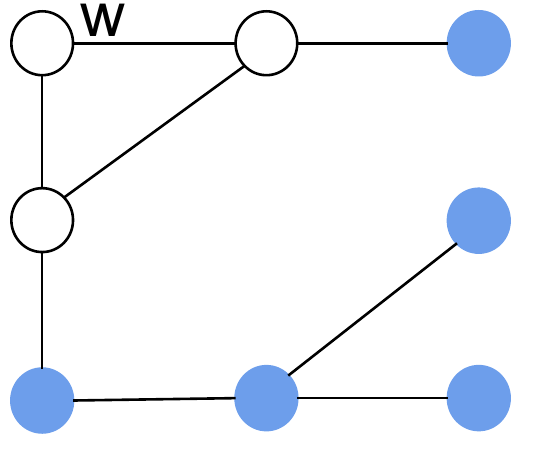}\label{fig:init_ex2}}
    \subfloat[Delete $v$ ]{\includegraphics[width=0.12\textwidth]{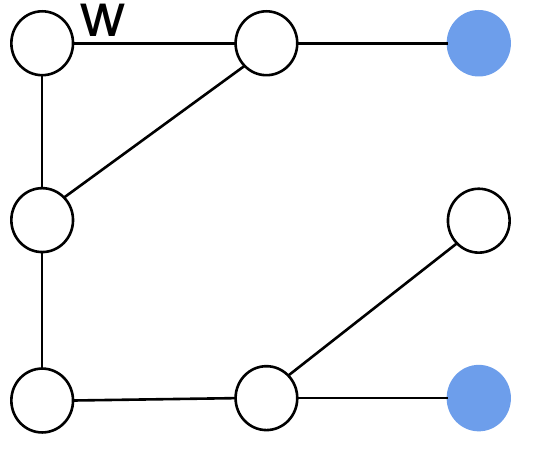}\label{fig:init_ex3}}
    \subfloat[Delete $w$ ]{\includegraphics[width=0.12\textwidth]{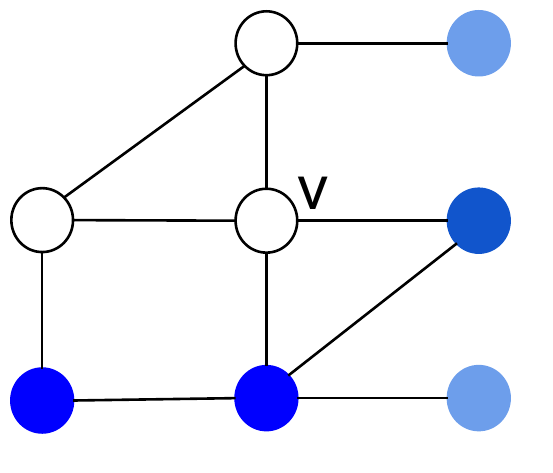}\label{fig:init_ex4}}
     
    \caption{\small(a) Initial graph: four nodes are in 3-core, three in 2-core and two nodes are in 1-core. (b) Considering just 3-core, deleting $a$ removes all other three nodes from 3-core. (c) In our problem \TMCV (Def. \ref{def:tmcv}), all the six empty nodes got affected after deleting $v$. (d)  In \TMCV, all three empty nodes got affected on deleting $w$.  \label{fig:tmcv_new_ex}}
\end{figure}

 We denote the modified graph $G$ after deleting a set $B$ consisting of $b$ vertices (nodes) as  $G\setminus B$. Deleting a vertex reduces the degree of its neighbours and possibly their coreness. This reduction in coreness might propagate to other vertices. Let us define an affected node as follows: a node $v$ is affected if $C(v,G)>C(v,G\setminus B)$. The example in Figure \ref{fig:init_ex3} shows that deleting a node (e.g. node $v$) can affect the neighbours and propagate to other non-neighbor nodes. Next we define the coreness minimization problem.
 
 \begin{definition} \label{def:tmcv} \textbf{Total Minimization of Coreness via Vertex deletion (\TMCV):} Given a graph $G\!=\!(V,E)$, candidate vertices $\Gamma\!\subseteq\!V$ and budget $b$, find the set $B\!\subset\!\Gamma$ of nodes to be removed such that $|B| \leq b$ and $f(B)\!=\!|\{v\!\in\!V\!\setminus\!B\!:\!C(v,G)\!>\!C(v,G\!\setminus\!B)\}|$ is maximized.
 \end{definition}

 Note that the objective minimizes the number of unaffected nodes. Intuitively, a network is more robust if its value of $f$ is small.
 \begin{example}
 Figure \ref{fig:init_graph} shows an example of the initial graph. The TMCV objective is explained in Figures \ref{fig:init_ex3} and \ref{fig:init_ex4}. In Figure \ref{fig:init_ex3}, when $v$ is deleted, all the remaining three nodes in the 3-core fall into 2-core and all the three nodes in the 2-core move to 1-core. Thus, five nodes are affected, i.e., $f(\{v\})=5$. In Figure \ref{fig:init_ex4}, by deleting the node $w$, only the nodes that were in the 3-core are affected, i.e., $f(\{w\})=3$. The empty nodes in Figures \ref{fig:init_ex2}, \ref{fig:init_ex3} and \ref{fig:init_ex4}  are the affected ones---i.e., with reduced coreness.
 \end{example}
 
\begin{figure*}[t]
    \centering
  {\includegraphics[width=0.8\textwidth]{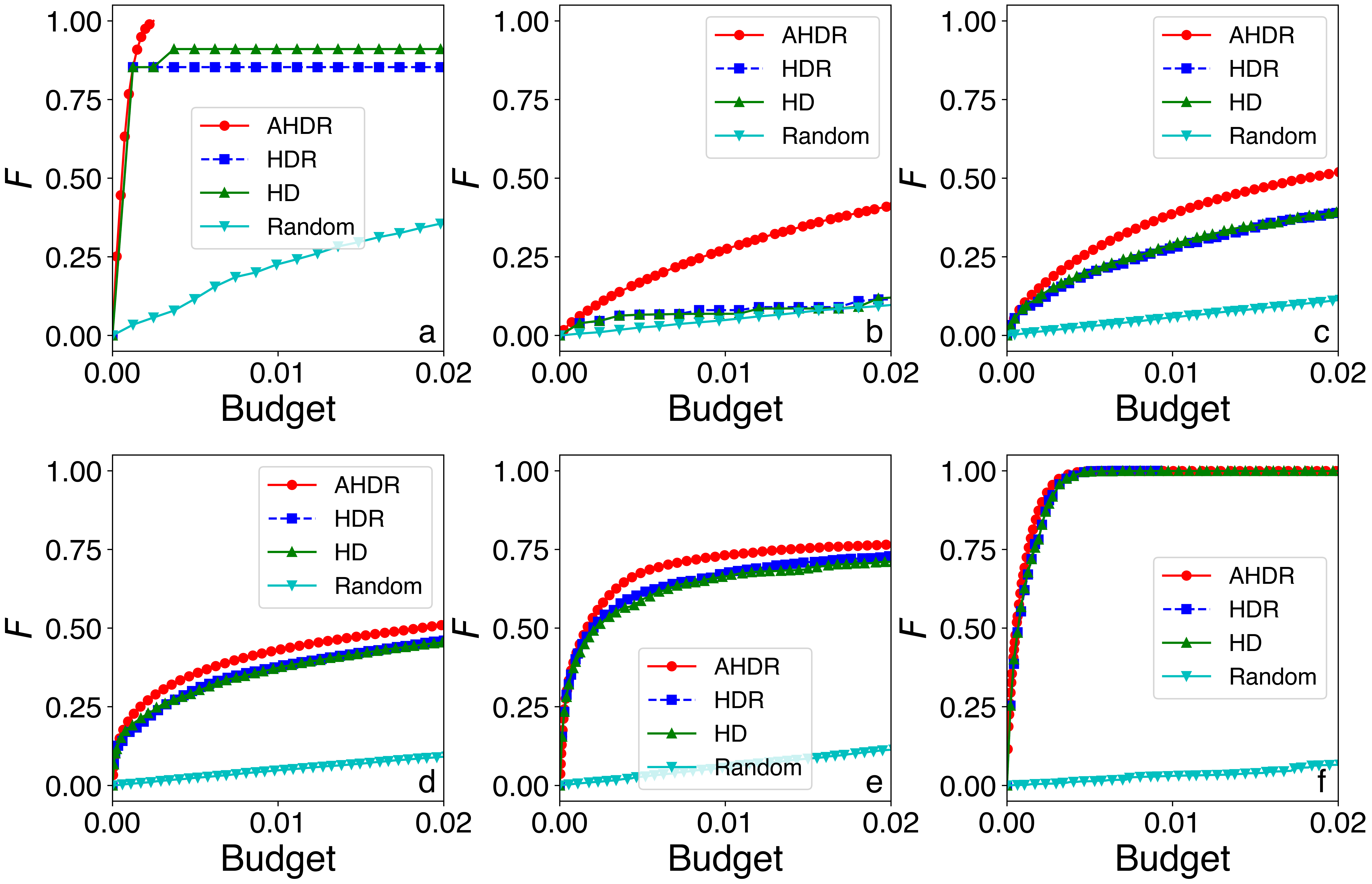}}

    \caption{ The performance of different heuristics varying the budget, $b$ (percentage of total nodes to delete) on real datasets: (a) Facebook, (b) GrQc, (c) CondMat, (d) BrightKite, (e) Enron, (f) g+.  \label{fig:algo_varying_budget}}
\end{figure*}

\section{Empirical Results} \label{sec:exp}
In this section, we motivate the \TMCV problem using real applications and simple heuristics. We show that our problem can be used to characterize different types of networks and to understand the relationship between protein-protein interaction (PPI) networks and the evolution of species.


\begin{table}[t]
	\centering
	
	\begin{tabular}{|c|c|c|c|c|}
	\hline
		\textbf{Dataset}& Type &$|V|$ & $|E|$ & $D$\\\hline
		Enron & Email & 36692 & 183831 & 43\\\hline
		GrQc & Co-authorship & 4158 & 13422 & 43\\\hline
CondMat & Co-authorship & 21363 & 91286 & 25\\\hline
		Facebook& Social & 4039 & 88234 & 115 \\\hline
g+ & Social & 23628 & 39194 & 12\\\hline

BrightKite & Social & 58228 & 214078 & 52\\\hline

	\end{tabular}
	\caption{Statistics of Datasets: $D$ denotes degeneracy, i.e., the maximum core.}\label{tab:dataset}
\end{table}

\subsection{Robustness and Characterization of Networks}
\label{sec:exp_budget}

We evaluate our robustness measure using different networks and show interesting properties of those via a few heuristics. We measure the performance of each algorithm by a disruption measure $F$ which is defined by the fraction of nodes getting affected (reduction of initial coreness) due to the deletion of the nodes in the solution set generated by each algorithm. A network is more robust if it has a lower value of $F$. We denote the modified graph $G$ as $\mathbb{G}_v^*$ after deleting a set $B$ consist of $b$ vertices (nodes). Formally, 
\begin{equation}
    F(B) =\frac{f(B)}{|V|}=\frac{|\{v\in V: C(v,G)>C(v,\mathbb{G}_v^*)\}|}{|V|}
\end{equation}

\paragraph{Datasets:} We use six real datasets from different genres in our experiments. Table \ref{tab:dataset} and Figure \ref{fig:core_disn_real} describe the statistics and the core distributions of the datasets, respectively. The datasets are available in \cite{nr-aaai15} and online\footnote{https://snap.stanford.edu/data}.

\subsubsection{Heuristics}
We describe the heuristics below.

Random: This algorithm chooses $b$ nodes randomly from the set of all nodes in the graph. The random strategy has been used in the past to enhance network robustness \cite{tang2015enhancing}.

High Degree (HD): It chooses top $b$ nodes according to their degree. Coreness is related to degree and the nodes in higher core usually contribute to the coreness in the lower core. So, this strategy uses degree as a proxy of the coreness. Intuitively, the algorithm should work well with the presence of sensitive nodes, i.e., when the degree of a node is equal to its individual coreness.

\begin{table}[h]
	\centering
	
	\begin{tabular}{|c|c|c|c|c|c|}
	\hline
		\textbf{Affected Size}& $0.1$ & $0.2$ & $0.3$ & $0.4$ & $0.5$\\\hline
		Enron & 3 & 7 & 16 & 33 & 61\\\hline
		GrQc & 10& 26 & 49 & 79 & 121\\\hline
CondMat & 22& 66 & 132 & 229& 385\\\hline
		Facebook& 1 & 1 & 2 & 2 & 3\\\hline
g+ & 1 & 3 & 6 & 9 & 14\\\hline

BrightKite & 10& 53& 169& 438 & 1076\\\hline

	\end{tabular}
	\caption{The effect of AHDR on different networks: The cell (Enron, 0.1) as 3 represents that only 3 nodes to be deleted to achieve $F= 0.1$, i.e., to affect 10\% of the Enron network. }\label{tab:AHDR_real}
\end{table}


High Disruption (HDR): The algorithm chooses top $b$ nodes according to their ``strength" in making nodes fall from their corresponding $k$-core. This strategy is more related with our objective function compared to the random and degree based heuristics. However, it requires the computation of the ``strength" for each node. The running time of this algorithm is $O(n^2+nm)$. 

Adaptive High Disruption (AHDR): It chooses the best node in each iteration for the budget number of iterations. However, in each step one needs to recompute the strength of the nodes given that already chosen nodes are deleted from the graph. 
A naive implementation of this strategy would take $O(bn(n+m))$ time, where $b$ is the budget. However, we are able to optimize this approach based on a few observations.

\begin{observation}
Deletion of a node $v$ might reduce coreness of another node $u$ only when $C(v,G)\geq C(u,G)$. There will not be any effect in deleting $v$ on $u$ if $C(v,G)< C(u,G)$.
\end{observation}

\begin{observation}
Based on the previous observation, node $v$ can be pruned from the candidate set $\Gamma$ if $C(u,G)> C(v,G), \forall{u\in N(v)}$ where $N(v)$ denotes the set of neighbors of $v$.
\end{observation}

\subsubsection{Results of different heuristics on real networks}

We vary the budget and evaluate the performance of each heuristic in all the datasets (Table \ref{tab:dataset}). Figure \ref{fig:algo_varying_budget} shows the results. Budget is the percentage of the total number of nodes in the network. A few interesting results are as follows: (a) AHDR is the most effective heuristic. The closest baseline, HDR, directly computes the effect of edge removal on the TMCV objective. AHDR, unlike others, considers the disruption in the network in an adaptive manner. (b) The efficacy of AHDR is more prominent in the co-authorship networks (CondMat and GrQc). The co-authorship networks often consist of small cliques and thus high degree (HD) or one shot strength computation (HDR) might not be effective to choose the critical nodes. 

Simple network properties such as density plays an important role in robustness. Facebook is a dense graph and the best heuristic, AHDR produces $F=1$ only with an extremely low budget (10 nodes). We further emphasize how robust the individual networks are by showing the number of nodes needed to be deleted by AHDR to affect a large portion of the network in Table \ref{tab:AHDR_real}. The dense structure makes the cores very sensitive and a node removal has high impact. Another interesting observation is that the graphs with the highest (Facebook) and the lowest (g+) densities are easy to disrupt compared to others. This suggests that the robustness does not entirely depend on the density of the graph.


\begin{figure*}[ht]
    \centering
    \subfloat[Real networks]{\includegraphics[width=0.26\textwidth]{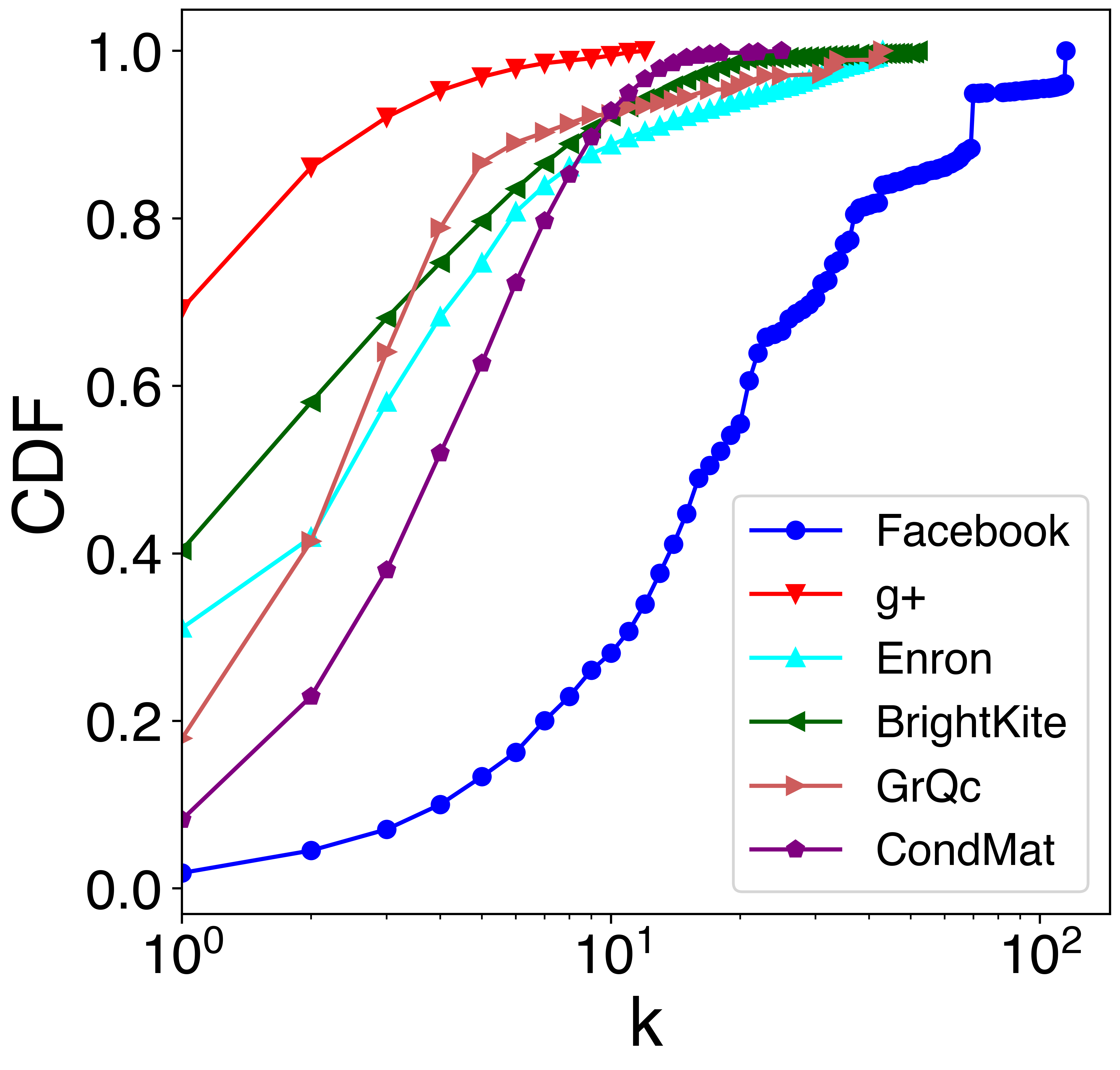} \label{fig:core_disn_real}}
    \subfloat[Networks used in Sec. \ref{sec:exp_syn}]{\includegraphics[width=0.26\textwidth]{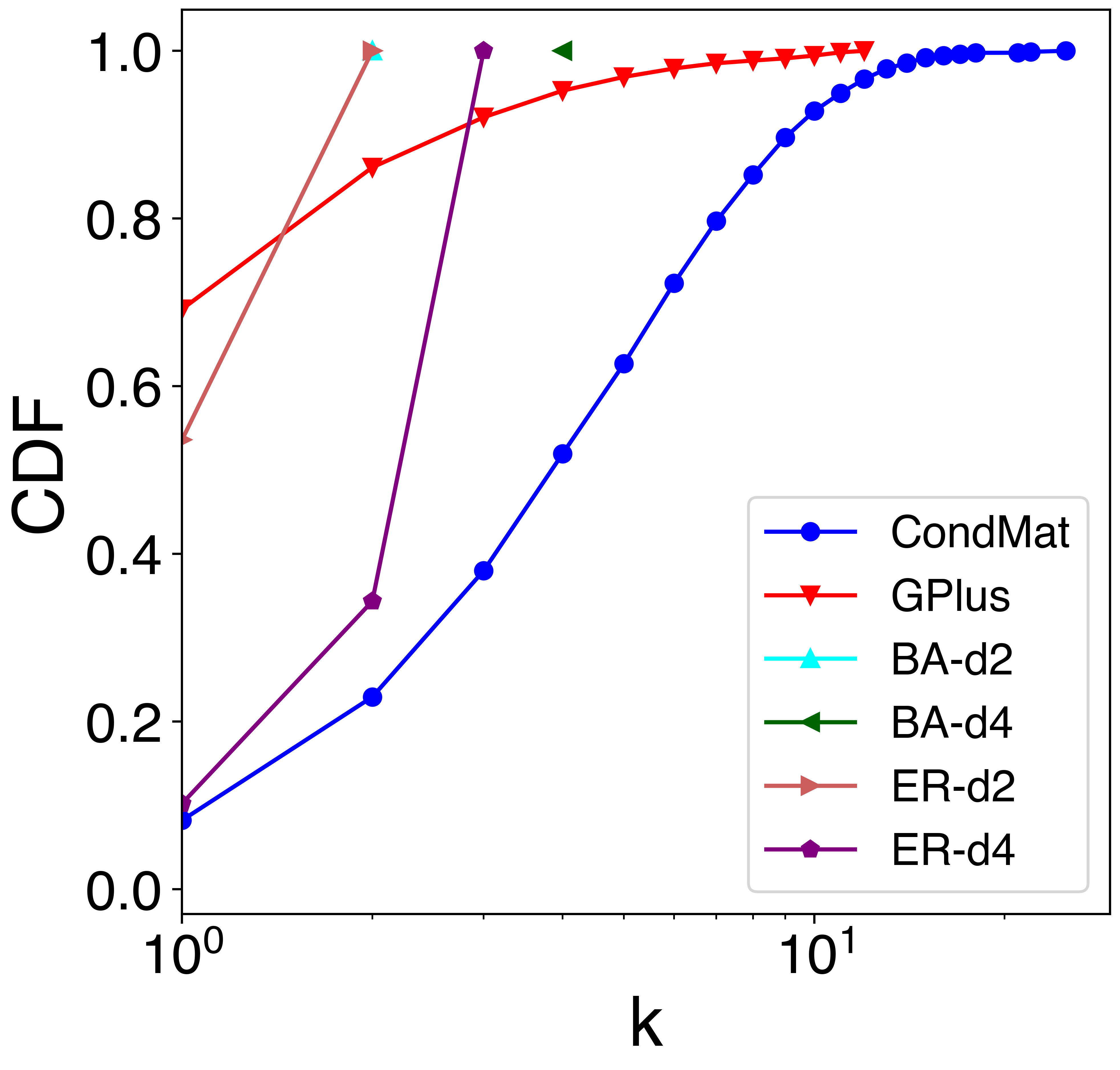}\label{fig:core_disn_syn}}
    \subfloat[Results by AHDR]{\includegraphics[width=0.27\textwidth]{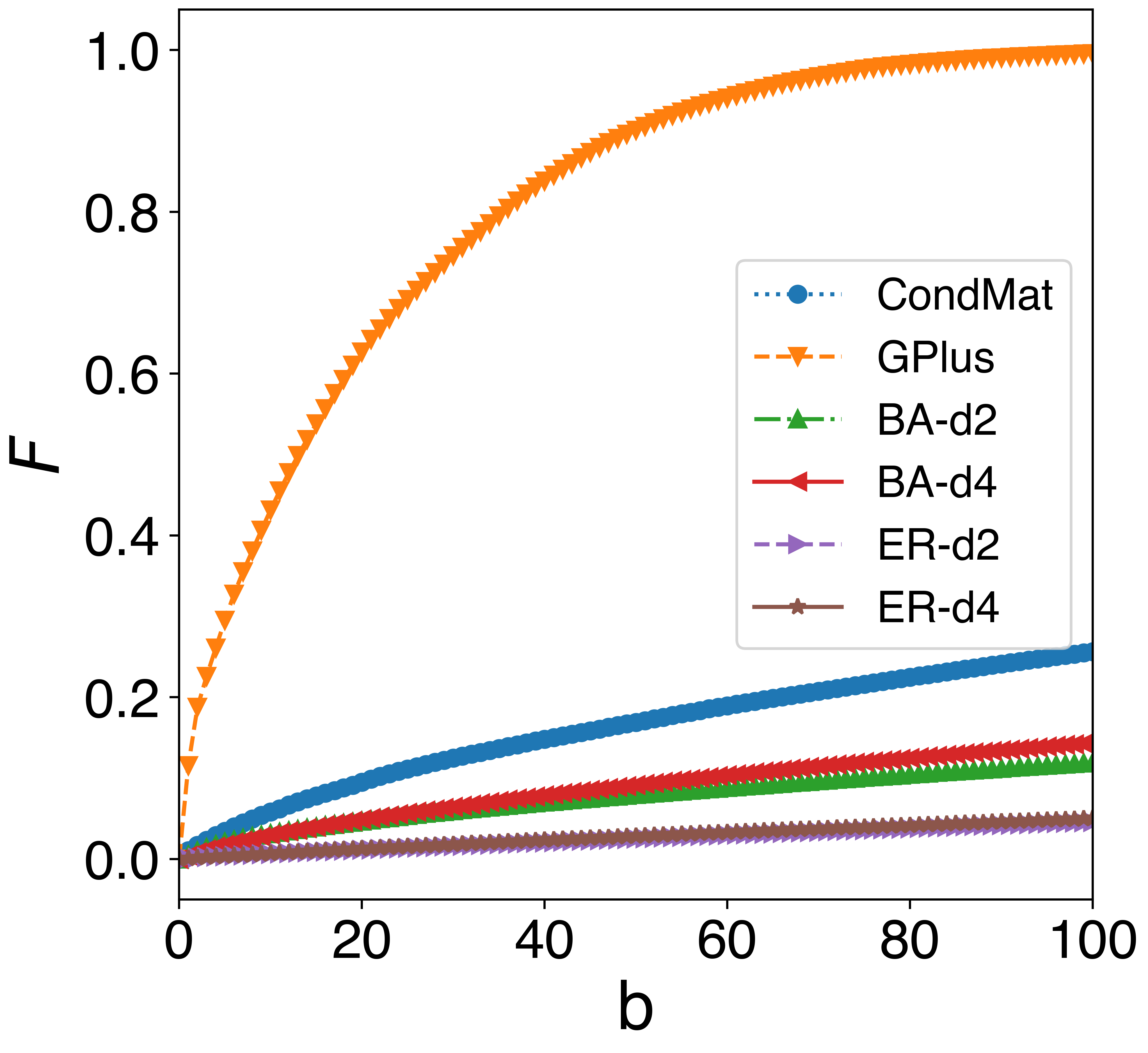}\label{fig:results_syn}}
    
    \caption{(a) The core distributions in the real networks in Table \ref{tab:dataset}. (b) The core distributions of the real and synthetic networks used in the experiments in Sec. \ref{sec:exp_syn}. (c) The performance of the best heuristic (AHDR) varying the number (budget, $b$) of nodes in co-authorship, social and synthetic networks of similar sizes.  \label{fig:synthetic_real}}
    
\end{figure*}

\subsubsection{Synthetic vs real networks} \label{sec:exp_syn} 

Figure \ref{fig:results_syn} shows the impact of the best performing heuristic, AHDR, in co-authorship (CondMat), social (g+) and synthetic ($|V|=20,000$) networks. BA-d2 (BA-d4) and ER-d2 (ER-d4) represent the synthetic network structures from two well-studied models: (a) Barabasi-Albert and (b) Erdos-Renyi, respectively, with average node degree $2$ (degree $4$). Note that all of these six networks have similar number of nodes. The $k$-core distributions of these networks is shown in Fig. \ref{fig:core_disn_syn}. The goal is to compare the robustness of different networks while applying the same algorithm (e.g., AHDR).
We observe that the random network, ER, is the most robust or difficult to break. As the edges are present uniformly across the network, node deletions do not have large affect on the network structure. This is true even with higher density ER graphs (see ER-d2 and ER-d4). Comparing BA and ER, BA is less robust to node removals as a few nodes have high degree and might be part of several cores. On the other hand, the real networks are less robust than both these synthetic networks. Even if the co-authorship network is denser than the social network (g+), the structure of g+ is less robust and ADHD can affect more than 80\% of the network by only removing 50 nodes. 

\subsection{Robustness and Evolution}
\label{sec:regression}


In the last section, we have applied network robustness as a tool to characterize different types of networks (email, co-authorship and social). Here, we use robustness to compare multiple networks of the same type.  Protein-protein interaction (PPI) networks capture how proteins interact to perform various biological functions (e.g., DNA replication, energy production). These networks are relevant in biological and biomedical applications, specially in the study of new treatments for complex diseases, such as cancer and autoimmune disorders \cite{safari2014protein}. Recently, it has been shown that the structure PPI networks is also related to the evolution of species \cite{Zitnik454033}. In particular, evolution was shown to be positively correlated with network resilience. In this section, we evaluate how k-core robustness can help us to better understand this relationship. 

\paragraph{Dataset:} For this study, we apply a subset of the Tree of Life dataset\footnote{http://snap.stanford.edu/tree-of-life}, which combines PPI networks and an evolution score---based on the depth in the phylogenetic tree---for 63 species. The species selected were those with at least 1,000 publications in the NCBI PubMed and belonging to the \textit{Bacteria} and \textit{Archaea} domains. 

\paragraph{Baseline:} We compare our resilience measure against the one applied in \cite{Zitnik454033}. More specifically, their approach measures how fragmented the network becomes after the removal of a fraction $\alpha$ of nodes selected at random. Once a node is removed, all its edges are also removed from the network. The fragmentation of $G_{\alpha}$ is measured based on a modified version of the Shannon divergence of the resulting connected components $\{V_1, V_2, \ldots V_K\}$:
\begin{equation*}
    H(G_{\alpha}) = \frac{1}{\log n} \sum_{k=1}^K p_k \log p_k
\end{equation*}
where $p_k=|V_k|/n$ and the $1/\log n$ factor enables comparing graphs with different sizes. 

The overall resilience of a network $G$ is computed as the area under the curve produced varying $\alpha$ from 0 to 1:
\begin{equation*}
    Resilience_{rand}(G) = 1 - \int_0^1 H(G_{\alpha})d\alpha
\end{equation*}

\paragraph{K-core Resilience:} We propose a resilience metric similar to the one defined above but replacing the Shannon entropy by the fraction of nodes out of their k-core:
\begin{equation*}
    Resilience_{core}(G) = 1 - \int_0^1 F(B_{\alpha})d\alpha
\end{equation*}
where $F(B_{\alpha})$ is the fraction of nodes affected after $\alpha|V|$ nodes are removed from $G$.

Similar to \cite{Zitnik454033}, we also apply our measures only to the largest connected component of each network. Moreover, we emphasize two key differences between our resilience metric ($Resilience_{core}$) and $Resilience_{rand}$. First, ours takes into account the k-core instead of the connected components in the graph. Second, we do not remove nodes at random, but as to maximize $F(B_{\alpha})$.




Figure \ref{fig::zitnik_resilience} shows the correlation between $Resilience_{rand}$ and the evolution of species. Notice that the measures have a weak correlation, with a Pearson's coefficient of $0.0325$ and a p-value of $0.80$. As a consequence, we are unable to reject the hypothesis that the variables are in fact uncorrelated. Notice that we consider a subset of the species from \cite{Zitnik454033}---with only the domains \textit{Bacteria} and \textit{Archaea}. Still, one would expect the correlation between evolution and resilience to also hold within these domains.  

\begin{figure}[h]
    \centering
    {\includegraphics[width=0.46\textwidth]{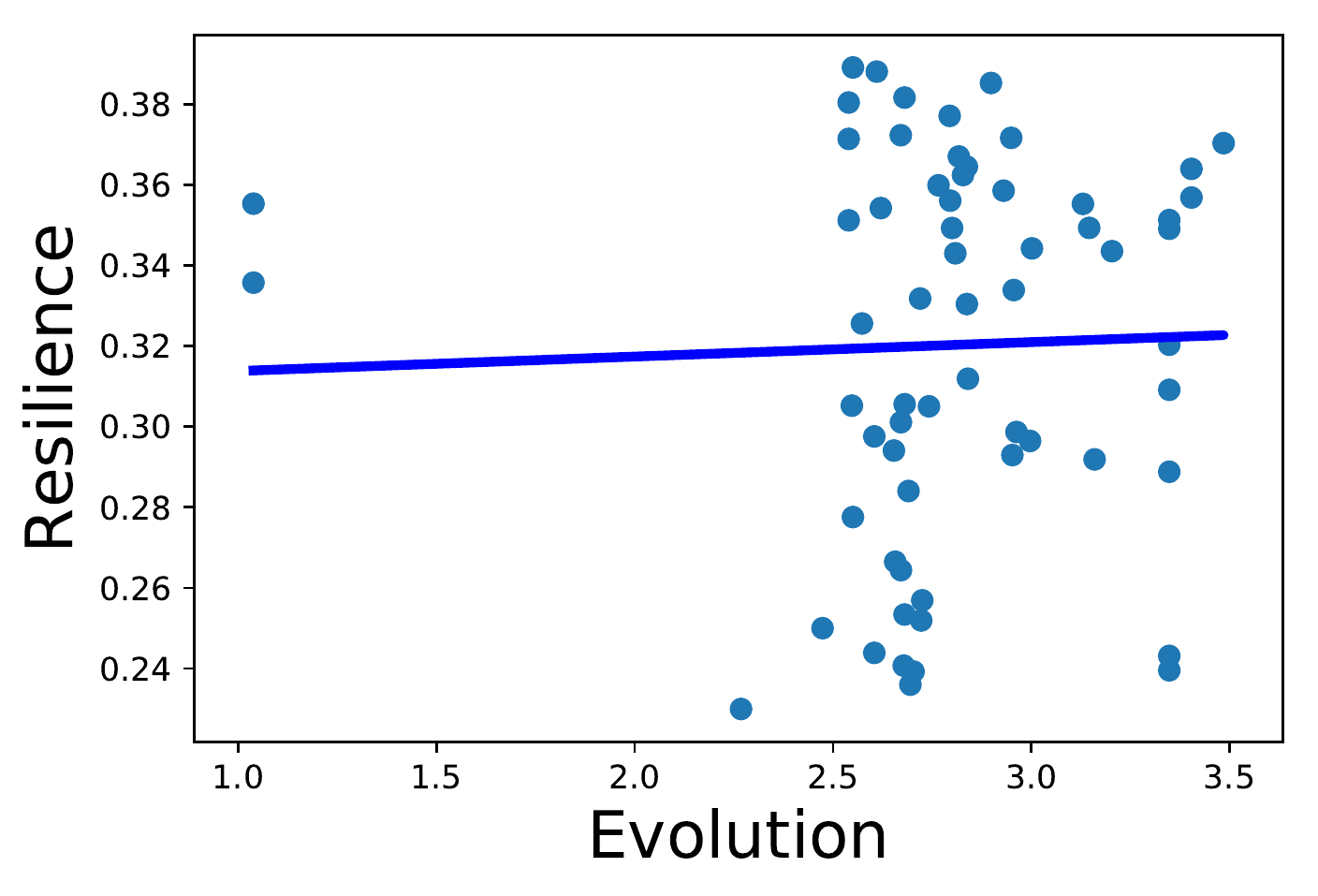}}
     
    \caption{Correlation (coefficient=$\mathbf{0.0325}$,p-value=$0.80$) between the evolution of species (x) and resilience (y) measured using $Resilience_{rand}$, which is based on random node removals \cite{Zitnik454033}. The plot and correlation values show that there is not a strong correlation between the measures. \label{fig::zitnik_resilience}}
\end{figure}

\begin{figure}[h]
    \centering
    {\includegraphics[width=0.46\textwidth]{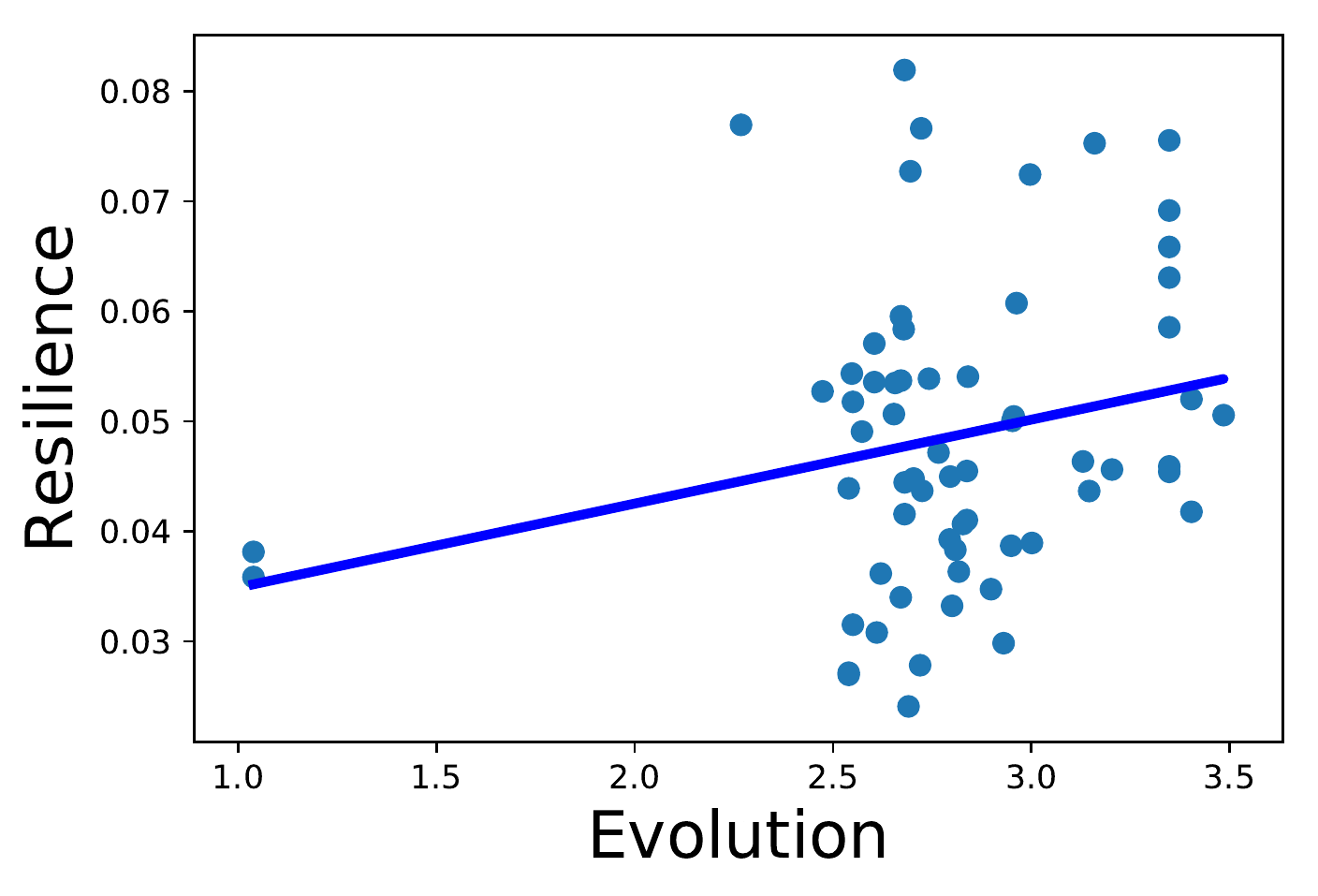}}
     
    \caption{Correlation (coefficient=$0.2366$,p-value=$0.06$) between evolution of species (x) and our resilience (y) measure based on k-cores. Our approach shows a significantly stronger correlation between the measures, which is an evidence that evolution induces PPI networks with a more resilient k-core structure. \label{fig::sourav_resilience}}
\end{figure}

In Figure \ref{fig::sourav_resilience}, we show the correlation between $Resilience_{core}$ (our metric) and evolution. Compared to Figure \ref{fig::zitnik_resilience}, we notice that our resilience measure has a stronger correlation with evolution of the species. In particular, the Pearson's coefficient for the correlation is 0.2366 with a small p-value of 0.06. This is a strong evidence that our notion of k-core resilience is able to capture relevant structural properties of PPI networks. Species that are further (or deeper) in the tree of life present a more robust network. More importantly, this relationship is even stronger when we consider a targeted attack, instead of random, to the k-core structure of the network. 

\section{Theoretical Results} \label{sec:theory}

The evaluation of our robustness measure relied on simple heuristics. However the question of finding an optimal algorithm still needs to be addressed.  In this section, we evaluate the hardness of the TMCV problem. The theoretical results show that there is no polynomial time algorithm even to achieve a constant factor approximation for the TMCV problem. From a parameterized perspective, we show that there is no fixed parameter tractable (FPT) algorithm for a few natural parameters such as the degeneracy, budget and the maximum degree of a node. The TMCV problem is either para-NP-hard or $W[2]$-hard for these parameters. However, for the size of the candidate set, there exists a FPT algorithm. We summarize our main results in Table \ref{tab:tmcv_summary}. 
 \subsection{Parameterized Complexity Results}

\begin{table}[t]
	\centering
	\renewcommand{\arraystretch}{1.2}
	\begin{tabular}{|c|c|}\hline
		 \textbf{Cond./Param.} & \textbf{Results}\\\hline
		
		  $b$& $W[2]$-hard (Theorem \ref{thm:TMCV_param_b}) \\\hline
		$D$ & para-NP-hard (Corollary \ref{cor:tmcv_para_d})\\\hline
		$\Delta$ &  para-NP-hard (Theorem \ref{thm:tmcv_para_deg}) \\\hline
		
		$|\Gamma|$ &  FPT (Observation \ref{obs:tmcv_param_gamma}) \\\hline
		
 $D(G)=1$ & Poly (Theorem \ref{lemma:np_hard_d_1_2})\\\hline
	 $D(G)\geq 3$ & NP-hard to approximate (Thm. \ref{thm:inapprox_tmcv})\\\hline
	\end{tabular}
	\caption{Summary of our hardness and parameterized complexity results for the TCMV problem. We denote the budget by $b$, the degeneracy (maximum coreness over all vertices) of the graph by $D(G)$ or $D$, the maximum degree of any vertex by $\Delta$, and the candidate set by $\Gamma$.}\label{tab:tmcv_summary}
\end{table}

Our first result shows that the TMCV problem is $W[2]$-hard parameterized by $b$. The proof involves an fpt-reduction from the well-known $W[2]$-hard \SC problem parameterized by the size of set cover~\cite{bonnet2016parameterized}. The \SC problem is defined as follows.

\begin{definition}[\SC]
	Given an universe \UU, a collection \SS of subsets of \UU, and a positive integer $r$, compute if there exists a subcollection $\WW\subseteq\SS$ such that (i) $|\WW|\le r$ and (ii) $\cup_{A\in\WW} A = \UU$.
\end{definition}
 
\begin{figure}[h]

    \centering
    {\includegraphics[width=0.28\textwidth]{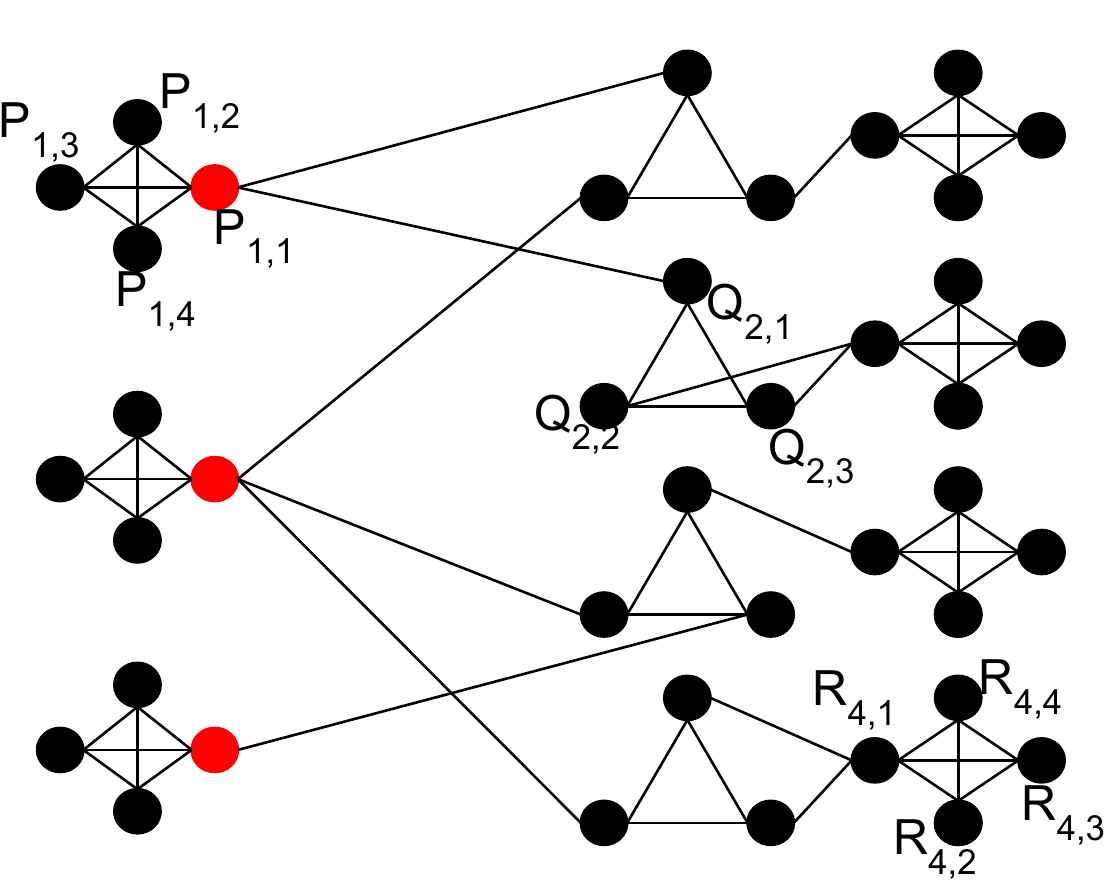}}
    \caption{An example of the construction for parameterized hardness of TMCV: reduction from Set Cover where  $U=\{u_1,u_2,u_3,u_4\}, S=\{S_1,S_2,S_3\}, S_1=\{u_1,u_2\},S_2=\{u_1,u_3,u_4\},S_3=\{u_3\}$. \label{fig:TCMV_param_b}}
  
\end{figure}

\begin{thm} \label{thm:TMCV_param_b}
The TMCV problem is $W[2]$-hard parameterized by $b$ for $k\geq 3$.
\end{thm}

\begin{proof}
Let $(\UU=\{ u_{1},u_{2},...,u_{n} \},\SS=\{S_{1},S_{2},...,S_{m}\}, r)$ be an instance of the \SC problem. We define a corresponding TMCV problem instance via constructing a graph $\GG$ as follows.

 For each $S_i \in \SS$ we create a clique of four vertices ($P_{i,1},\cdots,P_{i,4}$). For each $u_j \in \UU$, we create a cycle of $m$ vertices $Q_{j,1}, Q_{j,2},\cdots, Q_{j,m}$ with edges $(Q_{j,1},Q_{j,2}), \cdots,(Q_{j,m-1},Q_{j,m}), (Q_{j,m},Q_{j,1})$. 
We also create a clique of four vertices ($R_{j,1},\cdots,R_{j,4}$) for each $u_j \in \UU$. 
 Furthermore, edge $(P_{i,1},Q_{j,i})$ will be added to $\EE[\GG]$ if $u_j\in S_i$. Additionally, if $u_j\notin S_i$, edge $(Q_{j,i},R_{j,1})$ will be added to $\EE$. The candidate set, $\Gamma_v = \{ P_{i,1}| \forall{i=1,2,...,m}\}$.  Fig. \ref{fig:TCMV_param_b} illustrates our construction for sets $S_1=\{u_1,u_2\},S_2=\{u_1,u_3,u_4\}$ and $S_3=\{u_3\}$. We formally describe our TMCV instance as follows.
 \begin{align*}
     V[\GG] &= \{P_{i,t} : S_i\in \SS, t\in [4]\} \cup V_1, \text{where} \\
      V_1 &= \{Q_{j,i}: u_j\in \UU, S_i\in \SS\} \cup \{R_{j,t} : u_j\in \UU,  t\in[4]\} \\
      \EE[\GG] &= E_1\cup E_2\cup E_3 \cup E_4 \cup E_5, \text{where} \\
      E_1 &= \{(P_{i,s},P_{i,t}) :  i\in [m]; s,t \in [4], s\neq t\} \\
       E_2 &= \{(R_{j,s},R_{j,t}) : j\in [n]; s,t \in [4], s\neq t\}   \\
       E_3 &= \{(Q_{j,i},Q_{j,i+1}) : i\in [m-1], j\in [n]\} \\
      E_4 &=  \{(P_{i,1}, Q_{j,i})| u_j\in S_i, i \in [m], j\in [n]\} \\
      E_5 &= \{(Q_{j,i}, R_{j,1})|u_j\notin S_i, i \in [m], j\in [n]\} \\
      \Gamma &= \{P_{i,1}| i\in [m]\} \\
      b &= r
 \end{align*}
  
For any integer $z$, we denote the set $\{1,2,...,z\}$ by $[z]$. 
We now claim that the set cover instance is a \YES instance if and only if there exists a subset $B\subseteq\Gamma$ with $|B|\le b$ and $|\{v\in V\setminus B: C(v,\GG) > C(v,\GG\setminus B)\}|\ge 4b+mn$.

In one direction, let us assume that the \SC instance is a \YES instance. By renaming, let $S^\pr=\{S_1, \ldots, S_r\}$ form a set cover of the instance. We delete the nodes in the set $V^\pr=\{P_{i,1}| i\in[r]\}$ in the graph \GG. We claim that, by deleting the nodes in the set $V^\pr$, every node in $\{P_{i,t}| S_i \in S^\pr,t \in [4] \} \cup \{Q_{j,i}| j \in [n],i\in [m] \}$, i.e. $4b+mn$ nodes will go to the $2$-core.
We first observe that deletion of $P_{i,1}$ will move the other three nodes ($P_{i,2},P_{i,3},P_{i,4}$) to the $2$-core. 
Also, for any $j\in [n]$, if any connection $(Q_{j,i}, P_{i,1})$ gets deleted because of deletion of $P_{i,1}$, all the $m$ nodes in the set $\{Q_{j,t}| t\in [m]\}$ will go to $2$-core. Since $S^\pr$ forms a set cover for \UU, at least one connection $(Q_{j,l}, P_{l,1})$ where $j \in [n]$ and some $P_{l,1}\in V^\pr$ will get removed. Thus a total of $mn$ nodes will go to the $2$-core. Hence, the TMCV instance is a \YES instance.

For the other direction, we assume that there exists a subset $V^\pr=\{P_{i,1}|i\in [b]\}$ (by renaming) of size $b$ of the set $\Gamma$ such that deletion of $V^\pr$ would make at least $4b+nm$ nodes fall from the $3$-core. We claim that $S'=\{S_i: i \in [b]\}$ ($b=r$) forms a set cover for \UU. Suppose it is not, then at most $n-1$ connections among $(Q_{j,l}, P_{l,1})$ for $j \in [n]$ and some $P_{l,1}\in V^\pr$ will get deleted. Thus, at most $(n-1)m$ nodes will go into the $2$-core making it a total of $b+(n-1)m$ nodes falling from $3$-core. Hence, this is a contradiction and $S'$ is a set cover. 
\end{proof}

From the proof of Theorem \ref{thm:TMCV_param_b}, we obtain the following corollary. This follows from the observation that the \SC problem remains \NPC even if the size of each subset is $3$ and every element of the universe belongs to exactly $2$ subsets~\cite{DBLP:journals/arscom/FrickeHJ98}.

\begin{corollary}\label{cor:tmcv_para_d}
 The TMCV problem is para-NP-hard parameterized by $D$.
\end{corollary}

\begin{proof}
From the above construction, if we begin with an instance of \SC where the size of each subset is $3$ and every element of the universe belongs to exactly $2$ subsets, then we observe that the TMCV problem is also \NPH when $D=3$.
\end{proof}

We next consider maximum degree of the graph as parameter. By reducing from the Exact Cover problem~\cite{DBLP:books/daglib/0023376}, we show next the TMCV problem is para-NP-hard parameterized by the maximum degree of the input graph. The Exact Cover problem is the \SC problem where every set contains exactly $3$ elements from the universe. However, we use a special case of the Exact Cover problem where the elements are exactly in two subsets. This special case is known to be \NPC~\cite{DBLP:books/daglib/0023376}.

\begin{thm}
\label{thm:tmcv_para_deg}
The TMCV problem is para-NP-hard parameterized by the maximum degree ($\Delta$) in the graph.
\end{thm}

\begin{proof}
To prove our claim, we show a parameterized reduction from the special case of the Exact Cover problem which is known to be NP-hard. The problem is a \SC problem where each subset has exactly 3 elements and each element belongs to exactly two subsets. Let $(\UU=\{ u_{1},u_{2},...,u_{n} \},\SS=\{S_{1},S_{2},...,S_{m}\}, r)$ be an instance of the mentioned problem. We define a corresponding TMCV problem instance via constructing a graph $\GG$ as follows.

 We follow a similar reduction as in Theorem \ref{thm:TMCV_param_b}. For each $S_i \in \SS$ we create a clique of four vertices ($P_{i,1},\cdots,P_{i,4}$) for each $S_i \in \SS$. For each $u_j \in \UU$, we create two nodes $Q_{j,1}$ and $Q_{j,2}$ (we know each element belongs to exactly two subsets, one node is corresponding to the first subset and the second one is for the second in an arbitrary order) and an edge $(Q_{j,1},Q_{j,2})$ between them. We also create a clique of four vertices ($R_{j,1},\cdots,R_{j,4}$) for each $u_j \in \UU$. 
 Furthermore, edge $(P_{i,1},Q_{j,i})$ will be added to $\EE[\GG]$ if $u_j\in S_i$. Additionally, two edges $(Q_{j,1},R_{j,1})$ and $(Q_{j,2},R_{j,1})$ will be added to $\EE$. Clearly the reduction takes polynomial time. The candidate set, $\Gamma = \{ P_{i,1}| \forall{i=1,2,...,m}\}$. Note that the maximum degree in the graph is constant, i.e. $\Delta = 6$.
  
Initially in $\GG$, all vertices are in the $3$-core. We claim that a set $S'\subset S$, with $|S'|\leq r$, is a cover iff $f(B)=4b+2n$ where  $B\!=\!\{P_{i,1}| S_i \in S' \}$. 

Let us assume that the Exact Cover instance is a \YES instance and, by renaming, the collection $S^\pr=\{S_1, \ldots, S_r\}$ forms a valid set cover of the instance. We delete the nodes in the set $V^\pr=\{P_{i,1}| i\in[r]\}$ in the graph \GG. We claim that by deleting the nodes in the set $V^\pr$, every node in $\{P_{i,t}| S_i \in S^\pr,t \in [4] \} \cup \{Q_{j,i}| j \in [n],i\in \{1,2\} \}$, i.e. $4b+2n$ nodes will go in $2$-core.
We first observe that deletion of $P_{i,1}$ will make the other three nodes $P_{i,2}$ in $2$-core. Deletion of $b$ such nodes will lead $4b$ nodes falling into $2$-core.
Also, for any $j\in [n]$, if any connection $(Q_{j,i}, P_{i,1}$ gets deleted because of deletion of $P_{i,1}$, both nodes in the set $\{Q_{j,t}| t\in [2]\}$ will go to $2$-core. Since $S^\pr$ forms a set cover for \UU, at least one connection $(Q_{j,l}, P_{l,1})$ for all $j \in [n]$ and some $P_{l,1}\in V^\pr$ will get removed. Thus a total of $2m$ nodes will go in $2$-core. Hence, the TMCV instance is a \YES instance.

For the other direction, we assume that there exists a subset $V^\pr=\{P_{i,1}|i\in [b]\}$ (by renaming) of nodes of size $b$ of the set $\Gamma$ such that in the graph deletion of which would make at least $4b+2n$ nodes fall from the $3$-core. We claim that $S'=\{S_i: i \in [b]\}$ ($b=r$) forms a set cover for \UU. Suppose it is not, then at most $n-1$ connections among $(Q_{j,l}, P_{l,1})$ for $j \in [n]$ and some $P_{l,1}\in V^\pr$ will get deleted. Thus, at most $2(n-1)$ nodes will go into the $2$-core making it a total of $4b+2(n-1)$ nodes falling from $3$-core. Hence, this is a contradiction and $S'$ is a set cover. 
So, the TMCV problem is NP-hard when the maximum degree is constant ($\Delta=6$). Thus, the TMCV problem is para-NP-hard parameterized by the maximum degree ($\Delta$) in the graph.
\end{proof}

We conclude this section with the observation that the TMCV
problem is fixed parameter tractable parameterized by $|\Gamma|$. The algorithm simply tries all possible subsets of $|\Gamma|$ of size at most $b$.
\begin{observation}\label{obs:tmcv_param_gamma}
There is an algorithm for the TMCV problem running in time $O(2^{|\Gamma|} poly(n))$. In particular, TMCV is
fixed parameter tractable parameterized by $|\Gamma|$.
\end{observation}

\subsection{Inapproximability and Algorithm for $D(G)=1$}

In this section, we discuss the traditional hardness spectrum of the TMCV problem. We show a strong inapproximability result---it is NP-hard to achieve any $m^{-l_1}n^{-l_2}$-factor approximation even when $D(G) \geq 3$ for any constants $l_1>1$ and $l_2>1$.

\begin{thm}
\label{thm:inapprox_tmcv}
The TCMV problem is NP-hard to approximate within any $m^{-l_1}n^{-l_2}$-factor approximation even when $D(G) \geq 3$ for any constants $l_1>1$ and $l_2>1$.
\end{thm}

\begin{proof}
To prove our claim, first let us consider a reduction from the \SC problem. Let $(\UU=\{ u_{1},u_{2},...,u_{n} \},\SS=\{S_{1},S_{2},...,S_{m}\}, r)$ be an instance of the \SC problem. We define a corresponding TMCV problem instance via constructing a graph $\GG$ as follows.


We create a clique of four vertices ($P_{i,1},\cdots,P_{i,4}$) for each $S_i \in \SS$. 
For each $u_j \in \UU$, we create a cycle of $m$ vertices $Q_{j,1}, Q_{j,2},\cdots, Q_{j,m}$ with edges $(Q_{j,1},Q_{j,2}), \cdots,(Q_{j,m-1},Q_{j,m}), (Q_{j,m},Q_{j,1})$.

We also create a vertex $R$ along with a connected sub-graph on a set $\TT$ of $10 (m^{l_1}n^{l_2})^2$ vertices with degree exactly $3$ (for example, we can take a perfect matching between two cycles on $|\TT|/2$ vertices each). The node $R$ is attached with exactly two vertices in $\TT$. Furthermore, edge $(P_{i,1},Q_{j,i})$ will be added to $\EE[\GG]$ if $u_j\in S_i$. Additionally, if $u_j\notin S_i$, edge $(Q_{j,i},R)$ will be added to $\EE$. Clearly the reduction takes polynomial time. The candidate set, $\Gamma_v = \{ P_{i,1}| \forall{i=1,2,...,m}\}$.  Fig. \ref{fig:TCMV_inapprox} illustrates our construction for sets $S_1=\{u_1,u_2\},S_2=\{u_1,u_2,u_4\}$ and $S_3=\{u_3\}$. 

Initially in $\GG$, all vertices are in the $3$-core. We claim that a set $S'\subset S$, with $|S'|\leq r$, is a cover iff $f(B)= 3r + mn + 1 + |\TT|$ where  $B\!=\!\{P_{i,1}| S_i \in S' \}$. 

Let us assume that the \SC instance is a \YES instance and, by renaming, $S^\pr=\{S_1, \ldots, S_r\}$ forms a valid set cover of the instance. We delete the nodes in the set $V^\pr=\{P_{i,1}| i\in[r]\}$ in the graph $\GG$. We claim that by deleting the nodes in the set $V^\pr$, every node in $\{P_{i,t}| S_i \in S^\pr,t \in [4] \} \cup \{Q_{j,i}| j \in [n],i\in [m] \}$, i.e. $4b+mn$ nodes will go to the $2$-core. We first observe that deletion of $P_{i,1}$ will make the other three nodes $P_{i,2}$ to $2$-core. 
Also, for any $j\in [n]$, if any connection $(Q_{j,i}, P_{i,1})$ gets deleted because of deletion of $P_{i,1}$, all the $m$ nodes in the set $\{Q_{j,t}| t\in [m]\}$ will go to $2$-core. Since $S^\pr$ forms a set cover for $\UU$, at least one connection $(Q_{j,l}, P_{l,1})$ for all $j \in [n]$ and some $P_{l,1} \in V^\pr$ will get removed. Note that if
all the nodes $Q_{j,i}$;$\forall{j\in [n]}$ and $\forall{i\in [m]}$ go to $2$-core, the node $R$ will go to $2$-core and thus all the nodes in $\TT$ will follow the same. Thus a total of $3r+mn+1+|\TT|$ nodes will go to $2$-core. 

If there is no set cover of size $r$, then at most $n-1$ connections among $(Q_{j,l}, P_{l,1})$ for $j \in [n]$ and some $P_{l,1}\in V^\pr$ will get deleted. Thus, at most $(n-1)m$ nodes will go into the $2$-core making it a total of $3r+(n-1)m$ nodes falling from the $3$-core. Note that the node $R$ will be still in the $3$-core and thus the nodes in set $\TT$ will remain unaffected in the $3$-core. Hence, a total of $3r+(n-1)m$ nodes will go to $2$-core. 
So, the multiplicative difference of $f$s corresponding to the yes instance ($3r+mn+1+|\TT|$) and no instance ($3r+(n-1)m$) of the \SC problem is less than $m^{-l_1}n^{-l_2}$ and thus TMCV cannot be approximated within $m^{l_1}n^{l_2}$ factor unless $P\ne NP$. 
\end{proof}
\begin{figure}[t]

    \centering
    {\includegraphics[width=0.32\textwidth]{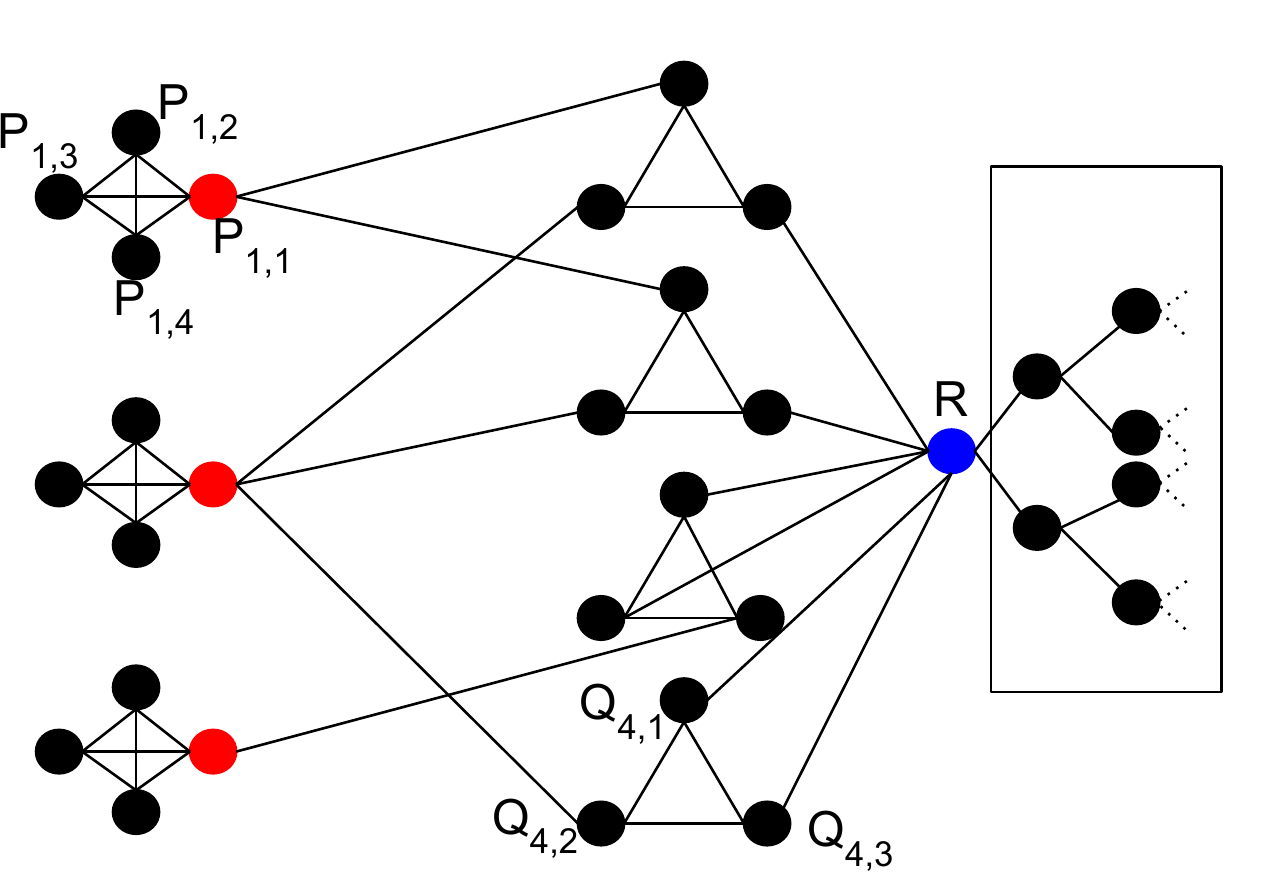}}
    \caption{Example of construction to prove inapproximability of TMCV: reduction from Set Cover where  $U=\{u_1,u_2,u_3,u_4\}, S=\{S_1,S_2,S_3\}, S_1=\{u_1,u_2\},S_2=\{u_1,u_2,u_4\},S_3=\{u_3\}$. The number of nodes in the rectangular box is $10 (m^{l_1}n^{l_2})^2$ and the nodes have exactly degree $3$. \label{fig:TCMV_inapprox}}
  
\end{figure}


We show next that the TCMV problem is polynomial time solvable if the degeneracy of the input graph is $1$.

\begin{thm} \label{lemma:np_hard_d_1_2}
The TCMV problem is polynomial time solvable if $D(v,G) =1$.
\end{thm}
\begin{proof}
Let $(G=(V,E),\Gamma\subseteq V, b)$ be any instance of TMCV such that $D(G)=1$. Since $D(G)=1$, it follows that $G$ is a forest. Let $G=T_1 \cup \cdots \cup T_k$ for some integer $k$ where $T_i$ is a tree for every $i\in[k]$. We first describe a dynamic programming based algorithm for the TCMV problem that works for trees.

Let $T$ be the input tree and $\Gamma_T\subseteq T$ the subset of vertices which can be deleted. We make the tree rooted at any node $r\in T$. At every node $x\in T$ and every integer $\el\in[b]\cup\{0\}$, we store the following.
\begin{align*}
    A[x,\el] &= \text{maximum number of isolated vertices in } \Gamma_T\cap T_x\\
    & \text{by deleting at most } \el \text{ vertices from $\Gamma_T\cap T_x$}\\
    & \text{subject to the condition that $x$ becomes isolated}\\
    B[x,\el] &= \text{maximum number of isolated vertices in } \Gamma_T\cap T_x\\
    & \text{by deleting at most } \el \text{ vertices from $\Gamma_T\cap T_x$ subject to}\\
    & \text{the condition that $x$ is neither isolated nor deleted}\\
    C[x,\el] &= \text{maximum number of isolated vertices in } \Gamma_T\cap T_x\\
    & \text{by deleting at most } \el \text{ vertices from $\Gamma_T\cap T_x$ subject to}\\
    & \text{the condition that $x$ is deleted}\\
    D[x,\el] &= \text{maximum number of isolated vertices in } \Gamma_T\cap T_x\\
    & \text{by deleting at most } \el \text{ vertices from $\Gamma_T\cap T_x$}
\end{align*}

From the definitions of $A[x,\el], B[x,\el], C[x,\el],$ and $D[x,\el]$, the following recurrences follow. Let the children and grandchildren of $x$ be respectively $y_1,\ldots,y_i$ and $z_1,\ldots,z_j$ ($j$ could be $0$).
\begin{align*}
    A[x,\el] &= \mathbbm{1}(\el\ge i) + \max\{D[z_1,\el_1]+\cdots+D[z_j,\el_j] :\\& \el_1+\cdots+\el_j \le \el-i \}\\
    B[x,\el] &= \max \bigcup_{\lambda=1}^i\mathlarger{\mathlarger{\{}} \max\{A[y_1,\el_1]-1, B[y_1,\el_1], C[y_1,\el_1]\}+\\
    \ldots&\max\{A[y_{\lambda-1},\el_{\lambda-1}]-1, B[y_{\lambda-1},\el_{\lambda-1}], C[y_{\lambda-1},\el_{\lambda-1}]\}\\
    &+\max\{A[y_\lambda,\el_\lambda]-1,B[y_\lambda,\el_\lambda]\}+\\
    &\max\{A[y_{\lambda+1},\el_{\lambda+1}]-1, B[y_{\lambda+1},\el_{\lambda+1}], C[y_{\lambda+1},\el_{\lambda+1}]\}\\
    +\ldots & +\max\{A[y_i,\el_i]-1, B[y_i,\el_i], C[y_i,\el_i]\}: \\
    &\el_1+\cdots+\el_i\le \el\mathlarger{\mathlarger{\}}}\\
    C[x,\el] &= \max\{D[y_1,\el_1]+\cdots+D[y_i,\el_i]:\\
    &\el_1+\cdots+\el_i\le\el-1\}\\
    D[x,\el] &= \max\{A[x,\el], B[x,\el], C[x,\el]\}
\end{align*}

We make the convention that the maximum over an empty set is $0$. For every leaf node $x$, we initialize $A[x,\el], B[x,\el], C[x,\el],$ and $D[x,\el]$ as follows.
\begin{align*}
    A[x,\el] &= 1, & \el\ge0\\
    B[x,\el] &= -1, & \el\ge0\\
    C[x,\el] &= \begin{cases}
    -1, & \el=0\\
    0, & \el>0
    \end{cases}\\
    D[x,\el] &= 1, &\el\ge0
\end{align*}

We observe that, given the tables at every descendant vertex of $x$, $A[x,\el]$ can be computed by a standard dynamic programming based algorithm for the knapsack problem in time $\OO(j)$~\cite{vazirani2013approximation}. Similarly, $B[x,\el]$ and $C[x\el]$ can be computed respectively in $\OO(i^2)$ and $\OO(i)$ time. Hence the running time of our algorithm is $\OO(n^2b)=\OO(n^3)$.
\end{proof}

\section{Conclusion}
In this work we have introduced a novel network robustness measure based on $k$-cores. More specifically, we have addressed the algorithmic problem that aims to maximize the number of nodes falling from their initial cores upon a given budget number of node deletions. We have characterized the hardness of the problem in both traditional and parameterized frameworks. Our problem is NP-hard to approximate by any constant, is $W[2]$-hard parameterized by the budget and is para-NP-hard for several other parameters such as degeneracy and maximum degree of the graph. We have also proposed a few heuristics and demonstrated their performance on several datasets. When applied to PPI networks, our approach has allowed us to correlate network resilience and the evolution of species. In the future, we will apply our resilience metric to the entire PPI database from \cite{Zitnik454033}. Moreover, we want to explore if there exist approximation algorithms for our problem in some relevant constrained cases beyond the ones considered here.


\bibliographystyle{ACM-Reference-Format} 
\bibliography{aamas21}


\begin{thebibliography}{32}


\ifx \showCODEN    \undefined \def \showCODEN     #1{\unskip}     \fi
\ifx \showDOI      \undefined \def \showDOI       #1{#1}\fi
\ifx \showISBNx    \undefined \def \showISBNx     #1{\unskip}     \fi
\ifx \showISBNxiii \undefined \def \showISBNxiii  #1{\unskip}     \fi
\ifx \showISSN     \undefined \def \showISSN      #1{\unskip}     \fi
\ifx \showLCCN     \undefined \def \showLCCN      #1{\unskip}     \fi
\ifx \shownote     \undefined \def \shownote      #1{#1}          \fi
\ifx \showarticletitle \undefined \def \showarticletitle #1{#1}   \fi
\ifx \showURL      \undefined \def \showURL       {\relax}        \fi
\providecommand\bibfield[2]{#2}
\providecommand\bibinfo[2]{#2}
\providecommand\natexlab[1]{#1}
\providecommand\showeprint[2][]{arXiv:#2}

\bibitem[\protect\citeauthoryear{Bhawalkar, Kleinberg, Lewi, Roughgarden, and
  Sharma}{Bhawalkar et~al\mbox{.}}{2015}]%
        {bhawalkar2015preventing}
\bibfield{author}{\bibinfo{person}{Kshipra Bhawalkar}, \bibinfo{person}{Jon
  Kleinberg}, \bibinfo{person}{Kevin Lewi}, \bibinfo{person}{Tim Roughgarden},
  {and} \bibinfo{person}{Aneesh Sharma}.} \bibinfo{year}{2015}\natexlab{}.
\newblock \showarticletitle{Preventing unraveling in social networks: the
  anchored k-core problem}.
\newblock \bibinfo{journal}{\emph{SIAM Journal on Discrete Mathematics}}
  \bibinfo{volume}{29}, \bibinfo{number}{3} (\bibinfo{year}{2015}),
  \bibinfo{pages}{1452--1475}.
\newblock


\bibitem[\protect\citeauthoryear{Bonnet, Paschos, and Sikora}{Bonnet
  et~al\mbox{.}}{2016}]%
        {bonnet2016parameterized}
\bibfield{author}{\bibinfo{person}{{\'E}douard Bonnet},
  \bibinfo{person}{Vangelis~Th Paschos}, {and} \bibinfo{person}{Florian
  Sikora}.} \bibinfo{year}{2016}\natexlab{}.
\newblock \showarticletitle{Parameterized exact and approximation algorithms
  for maximum k-set cover and related satisfiability problems}.
\newblock \bibinfo{journal}{\emph{RAIRO-Theoretical Informatics and
  Applications}} \bibinfo{volume}{50}, \bibinfo{number}{3}
  (\bibinfo{year}{2016}), \bibinfo{pages}{227--240}.
\newblock


\bibitem[\protect\citeauthoryear{Chitnis, Fomin, and Golovach}{Chitnis
  et~al\mbox{.}}{2013}]%
        {chitnis2013preventing}
\bibfield{author}{\bibinfo{person}{Rajesh~Hemant Chitnis},
  \bibinfo{person}{Fedor~V Fomin}, {and} \bibinfo{person}{Petr~A Golovach}.}
  \bibinfo{year}{2013}\natexlab{}.
\newblock \showarticletitle{Preventing Unraveling in Social Networks Gets
  Harder}. In \bibinfo{booktitle}{\emph{Twenty-Seventh AAAI Conference on
  Artificial Intelligence}}.
\newblock


\bibitem[\protect\citeauthoryear{Cormen, Leiserson, Rivest, and Stein}{Cormen
  et~al\mbox{.}}{2009}]%
        {DBLP:books/daglib/0023376}
\bibfield{author}{\bibinfo{person}{Thomas~H. Cormen},
  \bibinfo{person}{Charles~E. Leiserson}, \bibinfo{person}{Ronald~L. Rivest},
  {and} \bibinfo{person}{Clifford Stein}.} \bibinfo{year}{2009}\natexlab{}.
\newblock \bibinfo{booktitle}{\emph{Introduction to Algorithms, 3rd Edition}}.
\newblock \bibinfo{publisher}{{MIT} Press}.
\newblock
\showISBNx{978-0-262-03384-8}
\urldef\tempurl%
\url{http://mitpress.mit.edu/books/introduction-algorithms}
\showURL{%
\tempurl}


\bibitem[\protect\citeauthoryear{Crescenzi, D'Angelo, Severini, and
  Velaj}{Crescenzi et~al\mbox{.}}{2015}]%
        {crescenzi2015}
\bibfield{author}{\bibinfo{person}{Pierluigi Crescenzi},
  \bibinfo{person}{Gianlorenzo D'Angelo}, \bibinfo{person}{Lorenzo Severini},
  {and} \bibinfo{person}{Yllka Velaj}.} \bibinfo{year}{2015}\natexlab{}.
\newblock \showarticletitle{Greedily Improving Our Own Centrality in A
  Network}. In \bibinfo{booktitle}{\emph{SEA}}. \bibinfo{publisher}{Springer
  International Publishing}, \bibinfo{pages}{43--55}.
\newblock


\bibitem[\protect\citeauthoryear{Demaine and Zadimoghaddam}{Demaine and
  Zadimoghaddam}{2010}]%
        {demaine2010}
\bibfield{author}{\bibinfo{person}{Erik~D Demaine} {and}
  \bibinfo{person}{Morteza Zadimoghaddam}.} \bibinfo{year}{2010}\natexlab{}.
\newblock \showarticletitle{Minimizing the diameter of a network using shortcut
  edges}. In \bibinfo{booktitle}{\emph{Scandinavian Workshop on Algorithm
  Theory}}. Springer, \bibinfo{pages}{420--431}.
\newblock


\bibitem[\protect\citeauthoryear{Dilkina, Lai, and Gomes}{Dilkina
  et~al\mbox{.}}{2011}]%
        {dilkina2011}
\bibfield{author}{\bibinfo{person}{Bistra Dilkina},
  \bibinfo{person}{Katherine~J. Lai}, {and} \bibinfo{person}{Carla~P. Gomes}.}
  \bibinfo{year}{2011}\natexlab{}.
\newblock \showarticletitle{Upgrading shortest paths in networks}. In
  \bibinfo{booktitle}{\emph{Integration of AI and OR Techniques in Constraint
  Programming for Combinatorial Optimization Problems}}.
  \bibinfo{publisher}{Springer}, \bibinfo{pages}{76--91}.
\newblock


\bibitem[\protect\citeauthoryear{Ellens and Kooij}{Ellens and Kooij}{2013}]%
        {ellens2013graph}
\bibfield{author}{\bibinfo{person}{Wendy Ellens} {and}
  \bibinfo{person}{Robert~E Kooij}.} \bibinfo{year}{2013}\natexlab{}.
\newblock \showarticletitle{Graph measures and network robustness}.
\newblock \bibinfo{journal}{\emph{arXiv preprint arXiv:1311.5064}}
  (\bibinfo{year}{2013}).
\newblock


\bibitem[\protect\citeauthoryear{Fricke, Hedetniemi, and Jacobs}{Fricke
  et~al\mbox{.}}{1998}]%
        {DBLP:journals/arscom/FrickeHJ98}
\bibfield{author}{\bibinfo{person}{Gerd Fricke}, \bibinfo{person}{Stephen~T.
  Hedetniemi}, {and} \bibinfo{person}{David~Pokrass Jacobs}.}
  \bibinfo{year}{1998}\natexlab{}.
\newblock \showarticletitle{Independence and Irredundance in k-Regular Graphs}.
\newblock \bibinfo{journal}{\emph{Ars Comb.}}  \bibinfo{volume}{49}
  (\bibinfo{year}{1998}).
\newblock


\bibitem[\protect\citeauthoryear{Khalil, Dilkina, and Song}{Khalil
  et~al\mbox{.}}{2014}]%
        {Khalil2014}
\bibfield{author}{\bibinfo{person}{Elias~Boutros Khalil},
  \bibinfo{person}{Bistra Dilkina}, {and} \bibinfo{person}{Le Song}.}
  \bibinfo{year}{2014}\natexlab{}.
\newblock \showarticletitle{Scalable diffusion-aware optimization of network
  topology}. In \bibinfo{booktitle}{\emph{SIGKDD international conference on
  Knowledge discovery and data mining}}. ACM, \bibinfo{pages}{1226--1235}.
\newblock


\bibitem[\protect\citeauthoryear{Kimura, Saito, and Motoda}{Kimura
  et~al\mbox{.}}{2008}]%
        {kimura2008minimizing}
\bibfield{author}{\bibinfo{person}{Masahiro Kimura}, \bibinfo{person}{Kazumi
  Saito}, {and} \bibinfo{person}{Hiroshi Motoda}.}
  \bibinfo{year}{2008}\natexlab{}.
\newblock \showarticletitle{Minimizing the Spread of Contamination by Blocking
  Links in a Network.}. In \bibinfo{booktitle}{\emph{AAAI}}.
\newblock


\bibitem[\protect\citeauthoryear{Laishram, Sariy\"{u}ce, Eliassi-Rad, Pinar,
  and Soundarajan}{Laishram et~al\mbox{.}}{2018}]%
        {Laishram2018}
\bibfield{author}{\bibinfo{person}{Ricky Laishram},
  \bibinfo{person}{Ahmet~Erdem Sariy\"{u}ce}, \bibinfo{person}{Tina
  Eliassi-Rad}, \bibinfo{person}{Ali Pinar}, {and} \bibinfo{person}{Sucheta
  Soundarajan}.} \bibinfo{year}{2018}\natexlab{}.
\newblock \showarticletitle{Measuring and Improving the Core Resilience of
  Networks}. In \bibinfo{booktitle}{\emph{Proceedings of the 2018 World Wide
  Web Conference}}. \bibinfo{pages}{609--618}.
\newblock


\bibitem[\protect\citeauthoryear{Liu, Zhou, Wang, and Liu}{Liu
  et~al\mbox{.}}{2017}]%
        {liu2017comparative}
\bibfield{author}{\bibinfo{person}{Jing Liu}, \bibinfo{person}{Mingxing Zhou},
  \bibinfo{person}{Shuai Wang}, {and} \bibinfo{person}{Penghui Liu}.}
  \bibinfo{year}{2017}\natexlab{}.
\newblock \showarticletitle{A comparative study of network robustness
  measures}.
\newblock \bibinfo{journal}{\emph{Frontiers of Computer Science}}
  \bibinfo{volume}{11}, \bibinfo{number}{4} (\bibinfo{year}{2017}),
  \bibinfo{pages}{568--584}.
\newblock


\bibitem[\protect\citeauthoryear{Lordan and Albareda-Sambola}{Lordan and
  Albareda-Sambola}{2019}]%
        {lordan2019exact}
\bibfield{author}{\bibinfo{person}{Oriol Lordan} {and} \bibinfo{person}{Maria
  Albareda-Sambola}.} \bibinfo{year}{2019}\natexlab{}.
\newblock \showarticletitle{Exact calculation of network robustness}.
\newblock \bibinfo{journal}{\emph{Reliability Engineering \& System Safety}}
  \bibinfo{volume}{183} (\bibinfo{year}{2019}), \bibinfo{pages}{276--280}.
\newblock


\bibitem[\protect\citeauthoryear{Luo, Molter, and Suchy}{Luo
  et~al\mbox{.}}{2018}]%
        {luo2018parameterized}
\bibfield{author}{\bibinfo{person}{Junjie Luo}, \bibinfo{person}{Hendrik
  Molter}, {and} \bibinfo{person}{Ondrej Suchy}.}
  \bibinfo{year}{2018}\natexlab{}.
\newblock \showarticletitle{A Parameterized Complexity View on Collapsing
  k-Cores}.
\newblock \bibinfo{journal}{\emph{arXiv preprint arXiv:1805.12453}}
  (\bibinfo{year}{2018}).
\newblock


\bibitem[\protect\citeauthoryear{Medya, Bogdanov, and Singh}{Medya
  et~al\mbox{.}}{2018a}]%
        {medya2018making}
\bibfield{author}{\bibinfo{person}{Sourav Medya}, \bibinfo{person}{Petko
  Bogdanov}, {and} \bibinfo{person}{Ambuj Singh}.}
  \bibinfo{year}{2018}\natexlab{a}.
\newblock \showarticletitle{Making a small world smaller: Path optimization in
  networks}.
\newblock \bibinfo{journal}{\emph{IEEE Transactions on Knowledge and Data
  Engineering}} \bibinfo{volume}{30}, \bibinfo{number}{8}
  (\bibinfo{year}{2018}), \bibinfo{pages}{1533--1546}.
\newblock


\bibitem[\protect\citeauthoryear{Medya, Ma, Silva, and Singh}{Medya
  et~al\mbox{.}}{2020a}]%
        {medya2019k}
\bibfield{author}{\bibinfo{person}{Sourav Medya}, \bibinfo{person}{Tiyani Ma},
  \bibinfo{person}{Arlei Silva}, {and} \bibinfo{person}{Ambuj Singh}.}
  \bibinfo{year}{2020}\natexlab{a}.
\newblock \showarticletitle{A game theoretic approach for core resilience}.
\newblock \bibinfo{journal}{\emph{Proceedings of the 29th International Joint
  Conference on Artificial Intelligence}} (\bibinfo{year}{2020}).
\newblock


\bibitem[\protect\citeauthoryear{Medya, Silva, and Singh}{Medya
  et~al\mbox{.}}{2020b}]%
        {medya2020approximate}
\bibfield{author}{\bibinfo{person}{Sourav Medya}, \bibinfo{person}{Arlei
  Silva}, {and} \bibinfo{person}{Ambuj Singh}.}
  \bibinfo{year}{2020}\natexlab{b}.
\newblock \showarticletitle{Approximate Algorithms for Data-driven Influence
  Limitation}.
\newblock \bibinfo{journal}{\emph{IEEE Transactions on Knowledge and Data
  Engineering}} (\bibinfo{year}{2020}).
\newblock


\bibitem[\protect\citeauthoryear{Medya, Silva, Singh, Basu, and Swami}{Medya
  et~al\mbox{.}}{2018b}]%
        {medya2018group}
\bibfield{author}{\bibinfo{person}{Sourav Medya}, \bibinfo{person}{Arlei
  Silva}, \bibinfo{person}{Ambuj Singh}, \bibinfo{person}{Prithwish Basu},
  {and} \bibinfo{person}{Ananthram Swami}.} \bibinfo{year}{2018}\natexlab{b}.
\newblock \showarticletitle{Group centrality maximization via network design}.
  In \bibinfo{booktitle}{\emph{SIAM International Conference on Data Mining}}.
  \bibinfo{pages}{126--134}.
\newblock


\bibitem[\protect\citeauthoryear{Medya, Vachery, Ranu, and Singh}{Medya
  et~al\mbox{.}}{2018c}]%
        {medya2018noticeable}
\bibfield{author}{\bibinfo{person}{Sourav Medya}, \bibinfo{person}{Jithin
  Vachery}, \bibinfo{person}{Sayan Ranu}, {and} \bibinfo{person}{Ambuj Singh}.}
  \bibinfo{year}{2018}\natexlab{c}.
\newblock \showarticletitle{Noticeable network delay minimization via node
  upgrades}.
\newblock \bibinfo{journal}{\emph{Proceedings of the VLDB Endowment}}
  \bibinfo{volume}{11}, \bibinfo{number}{9} (\bibinfo{year}{2018}),
  \bibinfo{pages}{988--1001}.
\newblock


\bibitem[\protect\citeauthoryear{Meyerson and Tagiku}{Meyerson and
  Tagiku}{2009}]%
        {meyerson2009}
\bibfield{author}{\bibinfo{person}{A. Meyerson} {and} \bibinfo{person}{B
  Tagiku}.} \bibinfo{year}{2009}\natexlab{}.
\newblock \showarticletitle{Minimizing average shortest path distances via
  shortcut edge addition}. In \bibinfo{booktitle}{\emph{APPROX-RANDOM, I.
  Dinur, K.Janson, J.Noar and J. D. P. Rolim Eds, Vol. 5687. Springer}}.
  \bibinfo{pages}{272--285}.
\newblock


\bibitem[\protect\citeauthoryear{Rossi and Ahmed}{Rossi and Ahmed}{2015}]%
        {nr-aaai15}
\bibfield{author}{\bibinfo{person}{Ryan~A. Rossi} {and}
  \bibinfo{person}{Nesreen~K. Ahmed}.} \bibinfo{year}{2015}\natexlab{}.
\newblock \showarticletitle{The Network Data Repository with Interactive Graph
  Analytics and Visualization}. In \bibinfo{booktitle}{\emph{Proceedings of the
  Twenty-Ninth AAAI Conference on Artificial Intelligence}}.
\newblock
\urldef\tempurl%
\url{http://networkrepository.com}
\showURL{%
\tempurl}


\bibitem[\protect\citeauthoryear{Safari-Alighiarloo, Taghizadeh,
  Rezaei-Tavirani, Goliaei, and Peyvandi}{Safari-Alighiarloo
  et~al\mbox{.}}{2014}]%
        {safari2014protein}
\bibfield{author}{\bibinfo{person}{Nahid Safari-Alighiarloo},
  \bibinfo{person}{Mohammad Taghizadeh}, \bibinfo{person}{Mostafa
  Rezaei-Tavirani}, \bibinfo{person}{Bahram Goliaei}, {and}
  \bibinfo{person}{Ali~Asghar Peyvandi}.} \bibinfo{year}{2014}\natexlab{}.
\newblock \showarticletitle{Protein-protein interaction networks (PPI) and
  complex diseases}.
\newblock \bibinfo{journal}{\emph{Gastroenterology and Hepatology from bed to
  bench}} \bibinfo{volume}{7}, \bibinfo{number}{1} (\bibinfo{year}{2014}),
  \bibinfo{pages}{17}.
\newblock


\bibitem[\protect\citeauthoryear{Seidman}{Seidman}{1983}]%
        {seidman1983network}
\bibfield{author}{\bibinfo{person}{Stephen~B Seidman}.}
  \bibinfo{year}{1983}\natexlab{}.
\newblock \showarticletitle{Network structure and minimum degree}.
\newblock \bibinfo{journal}{\emph{Social networks}} \bibinfo{volume}{5},
  \bibinfo{number}{3} (\bibinfo{year}{1983}), \bibinfo{pages}{269--287}.
\newblock


\bibitem[\protect\citeauthoryear{Tang, Liu, and Zhou}{Tang
  et~al\mbox{.}}{2015}]%
        {tang2015enhancing}
\bibfield{author}{\bibinfo{person}{Xianglong Tang}, \bibinfo{person}{Jing Liu},
  {and} \bibinfo{person}{Mingxing Zhou}.} \bibinfo{year}{2015}\natexlab{}.
\newblock \showarticletitle{Enhancing network robustness against targeted and
  random attacks using a memetic algorithm}.
\newblock \bibinfo{journal}{\emph{EPL (Europhysics Letters)}}
  \bibinfo{volume}{111}, \bibinfo{number}{3} (\bibinfo{year}{2015}),
  \bibinfo{pages}{38005}.
\newblock


\bibitem[\protect\citeauthoryear{Tong, Prakash, Eliassi-Rad, Faloutsos, and
  Faloutsos}{Tong et~al\mbox{.}}{2012}]%
        {Tong2012GML}
\bibfield{author}{\bibinfo{person}{Hanghang Tong}, \bibinfo{person}{B~Aditya
  Prakash}, \bibinfo{person}{Tina Eliassi-Rad}, \bibinfo{person}{Michalis
  Faloutsos}, {and} \bibinfo{person}{Christos Faloutsos}.}
  \bibinfo{year}{2012}\natexlab{}.
\newblock \showarticletitle{Gelling, and melting, large graphs by edge
  manipulation}. In \bibinfo{booktitle}{\emph{Proceedings of the 21st ACM
  international conference on Information and knowledge management}}. ACM,
  \bibinfo{pages}{245--254}.
\newblock


\bibitem[\protect\citeauthoryear{Vazirani}{Vazirani}{2013}]%
        {vazirani2013approximation}
\bibfield{author}{\bibinfo{person}{Vijay~V Vazirani}.}
  \bibinfo{year}{2013}\natexlab{}.
\newblock \bibinfo{booktitle}{\emph{Approximation algorithms}}.
\newblock \bibinfo{publisher}{Springer Science \& Business Media}.
\newblock


\bibitem[\protect\citeauthoryear{Zhang, Zhang, Qin, Zhang, and Lin}{Zhang
  et~al\mbox{.}}{2017}]%
        {zhang2017finding}
\bibfield{author}{\bibinfo{person}{Fan Zhang}, \bibinfo{person}{Ying Zhang},
  \bibinfo{person}{Lu Qin}, \bibinfo{person}{Wenjie Zhang}, {and}
  \bibinfo{person}{Xuemin Lin}.} \bibinfo{year}{2017}\natexlab{}.
\newblock \showarticletitle{Finding critical users for social network
  engagement: The collapsed k-core problem}. In
  \bibinfo{booktitle}{\emph{Thirty-First AAAI Conference on Artificial
  Intelligence}}.
\newblock


\bibitem[\protect\citeauthoryear{Zhou and Liu}{Zhou and Liu}{2016}]%
        {zhou2016two}
\bibfield{author}{\bibinfo{person}{Mingxing Zhou} {and} \bibinfo{person}{Jing
  Liu}.} \bibinfo{year}{2016}\natexlab{}.
\newblock \showarticletitle{A two-phase multiobjective evolutionary algorithm
  for enhancing the robustness of scale-free networks against multiple
  malicious attacks}.
\newblock \bibinfo{journal}{\emph{IEEE transactions on cybernetics}}
  \bibinfo{volume}{47}, \bibinfo{number}{2} (\bibinfo{year}{2016}),
  \bibinfo{pages}{539--552}.
\newblock


\bibitem[\protect\citeauthoryear{Zhou, Zhang, Lin, Zhang, and Chen}{Zhou
  et~al\mbox{.}}{2019}]%
        {zhou2019k}
\bibfield{author}{\bibinfo{person}{Zhongxin Zhou}, \bibinfo{person}{Fan Zhang},
  \bibinfo{person}{Xuemin Lin}, \bibinfo{person}{Wenjie Zhang}, {and}
  \bibinfo{person}{Chen Chen}.} \bibinfo{year}{2019}\natexlab{}.
\newblock \showarticletitle{K-Core Maximization: An Edge Addition Approach}. In
  \bibinfo{booktitle}{\emph{Proceedings of the 28th International Joint
  Conference on Artificial Intelligence}}. AAAI Press,
  \bibinfo{pages}{4867--4873}.
\newblock


\bibitem[\protect\citeauthoryear{Zhu, Chen, Wang, and Lin}{Zhu
  et~al\mbox{.}}{2018}]%
        {zhu2018k}
\bibfield{author}{\bibinfo{person}{Weijie Zhu}, \bibinfo{person}{Chen Chen},
  \bibinfo{person}{Xiaoyang Wang}, {and} \bibinfo{person}{Xuemin Lin}.}
  \bibinfo{year}{2018}\natexlab{}.
\newblock \showarticletitle{K-core Minimization: An Edge Manipulation
  Approach}. In \bibinfo{booktitle}{\emph{Proceedings of the 27th ACM
  International Conference on Information and Knowledge Management}}. ACM,
  \bibinfo{pages}{1667--1670}.
\newblock


\bibitem[\protect\citeauthoryear{Zitnik, Sosic, Feldman, and Leskovec}{Zitnik
  et~al\mbox{.}}{2019}]%
        {Zitnik454033}
\bibfield{author}{\bibinfo{person}{Marinka Zitnik}, \bibinfo{person}{Rok
  Sosic}, \bibinfo{person}{Marcus~W Feldman}, {and} \bibinfo{person}{Jure
  Leskovec}.} \bibinfo{year}{2019}\natexlab{}.
\newblock \showarticletitle{Evolution of resilience in protein interactomes
  across the tree of life}.
\newblock \bibinfo{journal}{\emph{bioRxiv}} (\bibinfo{year}{2019}).
\newblock


\end{thebibliography}

\end{document}